\pdfoutput=1
\documentclass{article}

\usepackage[dblblindworkshop, final]{neurips_2026}

\usepackage[T1]{fontenc}
\usepackage{booktabs}
\usepackage{microtype}
\usepackage{graphicx}
\usepackage{subcaption}
\usepackage{placeins}
\usepackage{amsmath}
\usepackage{amssymb}
\usepackage{amsthm}
\usepackage{hyperref}
\usepackage{tikz}
\usetikzlibrary{calc,shadows,arrows.meta}
\usepackage{tcolorbox}

\definecolor{Muted}{HTML}{898781}
\definecolor{Hair}{HTML}{E1E0D9}
\definecolor{SeriesBlue}{HTML}{2A78D6}
\newtcolorbox{contributionsbox}{
  colback=SeriesBlue!4, colframe=SeriesBlue!70,
  boxrule=0pt, leftrule=2.5pt,
  arc=1pt, left=8pt, right=8pt, top=5pt, bottom=5pt, boxsep=0pt,
  before skip=6pt, after skip=6pt}

\theoremstyle{definition}
\newtheorem{definition}{Definition}[section]
\newtheorem{assumption}[definition]{Assumption}

\theoremstyle{plain}
\newtheorem{theorem}[definition]{Theorem}
\newtheorem{proposition}[definition]{Proposition}

\newcommand{\K}{\mathcal{K}}
\newcommand{\ctop}{c^{\top}}
\newcommand{\cbot}{c^{\bot}}
\newcommand{\coord}{\kappa}
\newcommand{\rmax}{r_{\max}}
\newcommand{\Cfg}{\mathcal{C}}
\newcommand{\Prin}{\mathcal{P}}
\newcommand{\Reach}{\mathrm{Reach}}
\newcommand{\Sep}{\mathsf{Sep}}
\newcommand{\Ret}{\mathsf{Ret}}
\newcounter{subfig}[figure]

\newcommand{\sublabel}[1]{\refstepcounter{subfig}\label{#1}}
\newcommand{\Sepb}{\Sep^{\mathrm{beh}}}
\newcommand{\Sepr}{\Sep^{\mathrm{rep}}}
\newcommand{\Wsp}{\mathcal{W}}

\workshoptitle{Foundations of Language Model Security}
\title{Capability-Gated Language Models:\\ Security Composes, Utility Does Not}

\author{Patrikas Vanagas \\
  BPTI\\
  Vilnius, Lithuania \\
  {\small\texttt{patrikas.vanagas@bpti.eu}} \\
  \And
  Augustas Mačijauskas \\
  Independent Researcher \\
  Vilnius, Lithuania \\
  {\small\texttt{august.macijauskas@gmail.com}} \\
  \And
  Laurynas Lopata \\
  askEarth AG\\
  Zürich, Switzerland \\
  {\small\texttt{lopatalaurynas@gmail.com}} \\
}

\newcommand{\mechViolations}{383}
\newcommand{\mechCells}{2016}
\newcommand{\mechMaxViolation}{4.3}
\newcommand{\mechLost}{3}
\newcommand{\mechBoth}{92}
\newcommand{\mechSlope}{0.66}
\newcommand{\mechGapMin}{$-$2.8}
\newcommand{\mechGapMax}{$+$5.6}
\newcommand{\mechMeetPPL}{1.18}
\newcommand{\mechGapMinA}{$+$2.2}
\newcommand{\mechGapMaxA}{$+$5.9}

\makeatletter
\@ifundefined{MPscratchCnt}{\input{supp-pdf.mkii}}{}
\makeatother

\begin{document}

\maketitle

\vspace{-8pt}
\begin{abstract}
Deployed language model safeguards (safety fine-tuning, filtering,
unlearning) vary by principal only \textit{outside} the model weights:
filters are reconfigured, tiers are multiplied, and artefacts are
reissued; inside one set of weights every request meets the same
model configuration.
This motivates us to define
\emph{capability-gated deployment}: per-principal access control inside one
set of weights, whose configurations form a lattice --- meets accumulate a
principal's restrictions and joins pool a coalition's reach. We instantiate it
by sparse rank gating over an existing nested-factorisation mechanism,
guide profile search with one-pass attribution, and read every result once from
a pre-registered held-out split.
Security approximately composes: provably exactly at meets under a monotone-elicitation assumption we falsify pointwise. In two lineages the median held-out meet deepens suppression; the one effect surviving correction strengthens it. Utility does not: individually harmless profiles can compose to retention and fluency damage, and no compositional bound exists.
{Code: \href{https://github.com/patrikasvanagas/capability-gated-deployment}{\texttt{github.com/patrikasvanagas/capability-gated-deployment}}}

\end{abstract}
\vspace{-8pt}

\section{Introduction}

A deployed language model is one endpoint and one configuration: a capability withheld
from an anonymous user is withheld from a vetted defender in the same deployment.
Protection theory has held since the 1970s that authority should instead be indexed to a
principal, so that a right can be granted, narrowed, or revoked for one subject without
altering it for every other \citep{lampson1974protection, graham1972protection}. The
mismatch is newly consequential for language models: frontier capability is rising on a
steep, well-characterised trend \citep{kwa2026measuringaiabilitycomplete},
dangerous-capability evaluation has become a precondition of release
\citep{shevlane2023modelevaluationextremerisks, phuong2024evaluatingfrontiermodelsdangerous},
and the domains at issue -- biology, chemistry, offensive cybersecurity
\citep{li2024wmdpbenchmarkmeasuringreducing} -- are exactly those in which different
principals warrant different capability.

The safeguards that ship are global by construction: post-training alignment shapes one
output distribution for every request thereafter
\citep{ouyang2022traininglanguagemodelsfollow, bai2022constitutionalaiharmlessnessai,
rafailov2024directpreferenceoptimizationlanguage}; unlearning and concept erasure edit the
weights, but the edited weights are the single artefact every user receives
\citep{li2024wmdpbenchmarkmeasuringreducing, belrose2025leaceperfectlinearconcept,
meng2023locatingeditingfactualassociations}; and classifiers mediate every request against
a policy beside the weights -- reconfigurable per tenant, but only over what the weights
still emit. Per-principal variation is thus bought outside any one set of weights --
reconfigured filters, tiered and identity-verified access, restricted releases,
managed-access frameworks [\mbox{\citealp{carter2026managed}},
\mbox{\citealp{bloomfield2026biological}}] -- at the
granularity of whole accounts and whole models: each distinct authority is a distinct
filter policy over intact capability, or a distinct artefact to train, evaluate, and
maintain \citep{roland2026modularpretrainingenablesaccess}. The caveat is known --
filtering leaves the underlying computation reachable, and unlearning is repeatedly found
to suppress rather than erase, the knowledge recoverable by fine-tuning or probing
\citep{lucki2025adversarialperspectivemachineunlearning,
deeb2025unlearningmethodsremoveinformation}: accuracy at chance is evidence about
behaviour, not about what remains reachable.

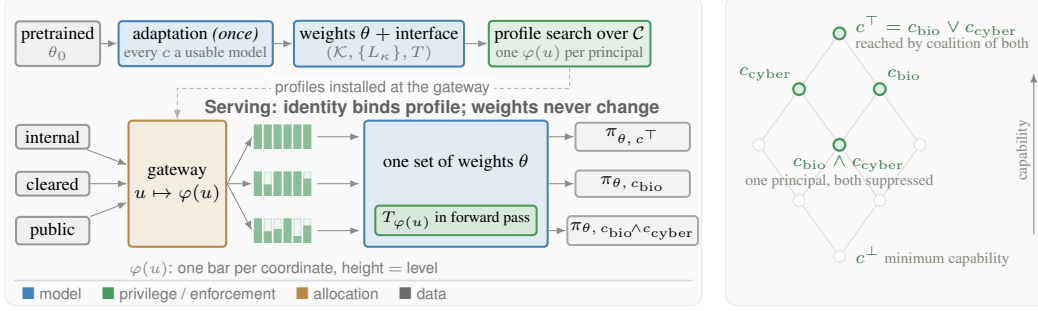
\begin{figure}[t]
\centering
\definecolor{colModel}{RGB}{70,130,180}
\definecolor{colPriv}{RGB}{76,153,96}
\definecolor{colSignal}{RGB}{180,140,70}
\definecolor{colNeutral}{RGB}{110,110,110}
\newcommand{\subl}[1]{{\fontsize{6}{7}\selectfont\color{black!55}#1}}
\newcommand{\notef}{\fontsize{6}{7}\selectfont}
\tikzset{hero/.style={
  x=1cm, y=1cm,
  blk/.style={rectangle, draw, rounded corners=1.5pt, font=\scriptsize, inner sep=2.2pt,
              align=center, thick},
  ttl/.style={font=\scriptsize\sffamily\bfseries, text=black!62},
  ann/.style={font=\notef\sffamily, text=black!52},
  ar/.style={-{Latex[length=1.3mm]}, draw=black!45, line width=0.4pt},
  arf/.style={-{Latex[length=1.3mm]}, draw=black!30, line width=0.4pt,
              dash pattern=on 1.6pt off 1.2pt},
  shad/.style={drop shadow={shadow xshift=0.4pt, shadow yshift=-0.4pt, opacity=0.12}},
  el/.style={circle, draw=Hair, fill=white, inner sep=1.6pt},
  mk/.style={circle, draw=colPriv, line width=0.9pt, fill=colPriv!18, inner sep=1.6pt},
  ed/.style={draw=Hair, line width=0.5pt},
  lab/.style={font=\scriptsize, text=colPriv!75!black, inner sep=1.5pt},
  note/.style={font=\notef, text=Muted, inner sep=1.2pt},
}}
\scalebox{0.97}{\begin{tikzpicture}[hero]
\path[use as bounding box] (-0.10,-2.12) rectangle (9.40,2.10);
\filldraw[fill=black!2, draw=black!10, line width=0.3pt, rounded corners=4pt] (-0.10,-2.12) rectangle (9.40,2.10);
\node[blk, fill=colNeutral!10, draw=colNeutral!70] (pre) at (0.62,1.45)
  {pretrained\\[-1pt]\subl{$\theta_0$}};
\node[blk, fill=colModel!16, draw=colModel, shad] (adapt) at (2.50,1.45)
  {adaptation \textit{(once)}\\[-1pt]\subl{every $c$ a usable model}};
\node[blk, fill=colModel!16, draw=colModel, shad] (wts) at (5.00,1.45)
  {weights $\theta$ $+$ interface\\[-1pt]\subl{$(\K,\{L_\coord\},T)$}};
\node[blk, fill=colPriv!16, draw=colPriv, shad] (srch) at (7.60,1.45)
  {profile search over $\Cfg$\\[-1pt]\subl{one $\varphi(u)$ per principal}};
\draw[ar] (pre) -- (adapt);  \draw[ar] (adapt) -- (wts);  \draw[ar] (wts) -- (srch);
\node[ttl, anchor=west] at (2.50,0.56) {Serving: identity binds profile; weights never change};
\foreach \y/\name/\cid in {0.18/internal/1, -0.46/cleared/2, -1.10/public/3} {
  \node[blk, fill=colNeutral!10, draw=colNeutral!70, minimum width=1.0cm] (p\cid) at (0.55,\y) {\name};
}
\node[blk, fill=colSignal!16, draw=colSignal, shad, minimum height=1.72cm] (gw) at (2.25,-0.46)
  {gateway\\[1pt]$u \mapsto \varphi(u)$};
\foreach \cid in {1,2,3} { \draw[ar] (p\cid) -- (gw); }
\foreach \y/\hh in {0.18/{1,1,1,1,1,1}, -0.46/{1,0.45,1,1,1,0.70}, -1.10/{1,0.45,0.55,1,0.25,0.70}} {
  \begin{scope}[shift={(3.30,\y)}]
    \foreach \h [count=\i from 0] in \hh {
      \draw[fill=colPriv!9, draw=colPriv!30, line width=0.2pt]
        (\i*0.135,-0.16) rectangle (\i*0.135+0.095,0.16);
      \draw[fill=colPriv!65, draw=colPriv!65, line width=0.2pt]
        (\i*0.135,-0.16) rectangle (\i*0.135+0.095,-0.16+\h*0.32);
    }
  \end{scope}
  \draw[ar] (gw.east) -- (3.27,\y);
  \draw[ar] (4.16,\y) -- (4.74,\y);
}
\node[ann, anchor=north, align=center] at (3.70,-1.40)
  {$\varphi(u)$: one bar per coordinate, height $=$ level};
\node[blk, fill=colModel!16, draw=colModel, shad, minimum height=1.72cm, minimum width=2.50cm]
  (mdl) at (6.05,-0.46) {};
\node[font=\scriptsize] at (6.05,-0.18) {one set of weights $\theta$};
\node[blk, fill=colPriv!22, draw=colPriv, minimum width=2.20cm, font=\notef] at (6.05,-0.96)
  {$T_{\varphi(u)}$ in forward pass};
\node[blk, fill=colNeutral!8, draw=colNeutral!60, minimum width=1.56cm, inner xsep=1pt] (o1) at (8.46,0.18)
  {$\pi_{\theta,\,\ctop}$};
\node[blk, fill=colNeutral!8, draw=colNeutral!60, minimum width=1.56cm, inner xsep=1pt] (o2) at (8.46,-0.46)
  {$\pi_{\theta,\,c_{\mathrm{bio}}}$};
\node[blk, fill=colNeutral!8, draw=colNeutral!60, minimum width=1.56cm, inner xsep=1pt] (o3) at (8.46,-1.10)
  {$\pi_{\theta,\,c_{\mathrm{bio}}\!\wedge c_{\mathrm{cyber}}}$};
\foreach \cid in {1,2,3} { \draw[ar] (mdl.east |- o\cid) -- (o\cid); }
\draw[arf] (srch.south) -- (7.60,0.84) -- (2.25,0.84) -- (gw.north);
\node[ann, anchor=center, fill=white, inner sep=1.5pt] at (5.05,0.84) {profiles installed at the gateway};
\draw[black!14, line width=0.6pt] (0.02,-1.78) -- (9.28,-1.78);
\node[ann, anchor=west] at (0.0,-1.98)
  {\textcolor{colModel}{$\blacksquare$}~model \quad
   \textcolor{colPriv}{$\blacksquare$}~privilege / enforcement \quad
   \textcolor{colSignal}{$\blacksquare$}~allocation \quad
   \textcolor{colNeutral}{$\blacksquare$}~data};
\end{tikzpicture}}\hspace{0.32cm}\scalebox{0.97}{\begin{tikzpicture}[hero]
\path[use as bounding box] (9.50,-2.12) rectangle (13.88,2.10);
\filldraw[fill=black!2, draw=black!10, line width=0.3pt, rounded corners=4pt] (9.50,-2.12) rectangle (13.88,2.10);
\begin{scope}[shift={(11.05,-1.45)}]
  \foreach \i in {0,1,2}{\foreach \j in {0,1,2}{
    \coordinate (n\i\j) at ({(\i-\j)*0.55},{(\i+\j)*0.76});}}
  \foreach \i/\j/\ii/\jj in {0/0/1/0, 0/0/0/1, 1/0/2/0, 1/0/1/1, 0/1/1/1, 0/1/0/2,
                             2/0/2/1, 1/1/2/1, 1/1/1/2, 0/2/1/2, 2/1/2/2, 1/2/2/2}
    {\draw[ed] (n\i\j) -- (n\ii\jj);}
  \foreach \p in {00,10,01,20,02} {\node[el] at (n\p) {};}
  \node[mk] (meet) at (n11) {};
  \node[mk] (cb)   at (n21) {};
  \node[mk] (cc)   at (n12) {};
  \node[mk] (top)  at (n22) {};
  \node[lab, anchor=west, align=left] at ($(top)+(0.16,0)$)
    {$\ctop = c_{\mathrm{bio}} \vee c_{\mathrm{cyber}}$\\[-2pt]{\notef\color{Muted}reached by coalition of both}};
  \node[lab, anchor=south west, inner sep=1pt] at ($(cb)+(0.06,0.08)$) {$c_{\mathrm{bio}}$};
  \node[lab, anchor=south east, inner sep=1pt] at ($(cc)+(-0.06,0.08)$) {$c_{\mathrm{cyber}}$};
  \node[lab, anchor=north, inner sep=0pt] (wedge) at ($(meet)+(0,-0.15)$) {$\wedge$};
  \node[lab, anchor=base east, inner sep=0pt] at ($(wedge.base west)+(-0.06,0)$) {$c_{\mathrm{bio}}$};
  \node[lab, anchor=base west, inner sep=0pt] at ($(wedge.base east)+(0.06,0)$) {$c_{\mathrm{cyber}}$};
  \node[note, anchor=north] at ($(meet)+(0,-0.35)$) {one principal, both suppressed};
  \node[lab, anchor=west] at ($(n00)+(0.16,0)$) {$\cbot$ {\notef\color{Muted}minimum capability}};
\end{scope}
\draw[ar, draw=Muted] (13.70,-1.15) -- (13.70,1.05);
\node[note, rotate=90, anchor=south] at (13.70,0) {capability};
\end{tikzpicture}}
\caption{\small \textbf{Capability-gated deployment.} Left (a): adaptation is paid once and yields
a single set of weights with the interface of Definition~\ref{def:interface}; a search over
$\Cfg$ then yields one profile per principal, which the gateway binds to identity and the
enforcer applies inside the forward pass. Right (b): profiles are points of a lattice: composing
two descends to their meet (one principal, both restrictions); pooling two principals'
outputs ascends to their join (a coalition's reach). Utility constraints in \eqref{eq:alloc}
bind at meets, security constraints at joins.}
\label{fig:hero}\sublabel{fig:hero-a}\sublabel{fig:hero-b}
\end{figure}

What is missing is thus not the will to differentiate but a per-principal mechanism inside
one set of weights and an algebra for it, since authority is useful
because it composes: a principal under two restrictions holds their conjunction, and
principals pooling outputs command their union. The shape of that algebra is classical --
Denning's lattice model combines security requirements by greatest lower bound
\citep{denning1976lattice}, and non-interference says what unavailability means
\citep{goguen1982security} -- and so is the debt: such properties are notoriously not
preserved under composition \citep{mccullough1988noninterference}, and which strengthened
conditions survive, at what price, is literature of its own
\citep{mclean1994general, mantel2002composition}. LLM capability profile lattice inherits both.

\begin{contributionsbox}
We define \emph{capability-gated deployment} (Fig.~\ref{fig:hero-a}): per-principal access
control inside one set of weights, varying only one privilege level per control coordinate.
Configurations form a finite distributive lattice: a principal under several restrictions
holds the meet, a coalition pooling outputs commands the join (Fig.~\ref{fig:hero-b}). We instantiate it by sparse rank gating over
a nested factorisation, steer profile search with a one-pass attribution, and read every
result once from a pre-registered held-out split --- an A/B discipline portable
to any capability-control evaluation. Security composes at meets (provably
under a monotone-elicitation assumption we then falsify pointwise, and empirically in two
lineages) while utility does not: individually harmless profiles compose to
real retention and fluency damage, and no compositional bound is possible.
\end{contributionsbox}

\textbf{Scope.} All claims assume query access (outputs observed at one's own
configuration, identity bound at the gateway, weights untouched) in expectation
over a declared elicitation distribution and against a declared probe class; the
gap from these expectation-based guarantees to worst-case,
differential-privacy-style separation is the open problem we pose.

\section{Capability-Gated Language Models}

What \emph{algebraic} object serves many principals from one deployment, and what survives
composition?

\subsection{Formal Definition of Capability-Gated Deployment}
\label{sec:formal}
\textbf{The Configuration Lattice.} The weights are fixed; deployment varies only \emph{where}, and \emph{how far},
privilege is removed. We fix that action space abstractly, so the algebra belongs to
the interface.

\begin{definition}[Capability-gating interface]\label{def:interface}
A \emph{capability-gating interface} for a model with parameters $\theta$ is a triple
$(\K, \{L_\coord\}_{\coord \in \K}, T)$, where $\K$ is a finite set of \emph{control
coordinates} (sites whose privilege may be set independently) and each $\coord \in \K$
carries a finite, totally ordered set $L_\coord$ of \emph{privilege levels}. A
\emph{configuration} assigns a level to every coordinate, i.e.\ is an element of
$\Cfg := \prod_{\coord \in \K} L_\coord$; we write $\ctop$ for the configuration granting the
greatest level everywhere and $\cbot$ for the least. The \emph{enforcer} $T$ carries a
configuration to effective parameters $\theta(c) := T_c(\theta)$ and thus to a deployed policy
$\pi_{\theta,c} := \pi_{\theta(c)}$, with $T_{\ctop}(\theta) = \theta$. Configurations are
ordered pointwise,
\begin{equation}
  c \preceq c' \quad :\Longleftrightarrow \quad c(\coord) \le c'(\coord)
  \ \ \text{for every } \coord \in \K ,
  \label{eq:order}
\end{equation}
and the interface is required to be \emph{monotone in reachability}: writing $\Reach(c)$ for
the set of computations realisable under $\theta(c)$,
\begin{equation}
  c \preceq c' \quad \Longrightarrow \quad \Reach(c) \subseteq \Reach(c') .
  \label{eq:mono}
\end{equation}
\end{definition}

Equation~\eqref{eq:mono} gives the order its content: ascending it can only enlarge what the
model may compute, and \eqref{eq:order} compares two configurations only when one dominates
the other at every coordinate. Under \eqref{eq:order}, $(\Cfg, \preceq)$ is a finite
distributive lattice, bounded by $\cbot$ and $\ctop$, with
$(c \wedge c')(\coord) = \min\{c(\coord), c'(\coord)\}$ and
$(c \vee c')(\coord) = \max\{c(\coord), c'(\coord)\}$ (Theorem~\ref{thm:lattice},
App.~\ref{app:proofs}; Fig.~\ref{fig:hero-b}).

\textbf{Principals, Profiles, and Allocation.} Deployment is not a model but a service: one set of weights answers to several parties whose
authorities differ. $\Prin$ is a finite set of \emph{principals} --- users, endpoints, or
API keys. Each $u \in \Prin$ carries a \emph{suppress set} $S_u$ of capability domains
that must be unavailable to $u$ and a \emph{retain set} $R_u$ of tasks on which $u$'s
utility must be preserved. \emph{Allocation} is a map $\varphi : \Prin \to \Cfg$;
$\varphi(u)$ is $u$'s \emph{profile}, bound to the principal at the gateway and enforced
inside by $T_{\varphi(u)}$ (Fig.~\ref{fig:hero-a}). Principals hold query access only:
each sees outputs at its own profile and alters neither it nor $\theta$.

\textbf{Two Phases.} The weights are adapted once; the profiles are found afterwards. \emph{Adaptation} is a
one-time procedure returning $\theta$ together with an interface satisfying
Definition~\ref{def:interface}, under which \emph{every} configuration is a usable model. \emph{Profile synthesis} then searches $\Cfg$, per principal, for a
configuration meeting that principal's requirements. Only the latter repeats when a
principal is added, and it is the hard one: $|\Cfg| = \prod_{\coord} |L_\coord|$ is
exponential in $|\K|$, so synthesis measures only a shortlist. It needs an
\emph{attribution} (effect of each level of each coordinate on each domain, predicted
from the mechanism's internals at cost independent of $|\Cfg|$) which the interface
does not provide and an instantiation must.

\begin{definition}[Capability-gated deployment]\label{def:cgd}
Fix an interface as in Definition~\ref{def:interface} and a family
$\mathfrak{T} \subseteq 2^{\Prin}$ of coalitions the threat model admits. Let $\Sep(c, D)$
be a predicate that domain $D$ is \emph{separated} under $c$, $\Ret(c, R)$ that task set
$R$ is \emph{retained}, and $\gamma_R(c)$ the \emph{collateral} incurred on $R$;
\S\ref{sec:modularity} gives them content. \emph{Capability-gated deployment} solves
\begin{equation}
\begin{aligned}
  \min_{\varphi \,:\, \Prin \to \Cfg} \ \ & \sum_{u \in \Prin} \gamma_{R_u}\big(\varphi(u)\big)
  \\[-1pt]
  \text{s.t.} \ \ & \Sep\big(\varphi(u), D\big) \ \ \forall u \in \Prin,\ \forall D \in S_u,
    \qquad \Ret\big(\varphi(u), R_u\big) \ \ \forall u \in \Prin, \\[-1pt]
  & \Sep\Big(\textstyle\bigvee_{u \in T} \varphi(u),\ D\Big) \ \ \forall T \in \mathfrak{T},
    \ \forall D \in \textstyle\bigcap_{u \in T} S_u .
\end{aligned}
\label{eq:alloc}
\end{equation}
\end{definition}

The two constraint families sit on opposite sides of the lattice (Fig.~\ref{fig:hero-b}).
A principal obliged to
forgo several domains at once holds the meet of the corresponding profiles, so its utility is
assessed \emph{below} each of them; a coalition pooling outputs commands whatever any member
commands, so its security must be assessed at the join, \emph{above} them all. Utility is
therefore evaluated at meets and security at joins --- an
asymmetry \eqref{eq:alloc} is arranged to expose rather than to hide.

\subsection{Security and Utility Under Composition}
\label{sec:modularity}

\begin{definition}[Channels, separation, retention]\label{def:seppred}
A configuration $c$ is observed at two channels. Its \emph{behavioural} channel is what
it emits, summarised on a domain $D$ with elicitation distribution $\mathcal{D}_D$ and
chance level $\alpha_D$ as
$\mathrm{acc}_D(c) = \mathbb{E}_{x \sim \mathcal{D}_D}\!\big[\mathrm{acc}_{\pi_c}(x)\big]$;
its \emph{representational} channel is the state $h_c(x)$ at the decision, read by a
declared class $P$ of probes from representations to answers. With
$\gamma_R(c) = \mathrm{acc}_R(\ctop) - \mathrm{acc}_R(c)$ the collateral on a retain set
$R$ and corresponding tolerances $\varepsilon, \delta, \tau \ge 0$,
\predisplaypenalty=0
\begin{equation}
\begin{aligned}
  \Sepb_\varepsilon(c, D) &\ :\Longleftrightarrow\
    \mathrm{acc}_D(c) \le \alpha_D + \varepsilon
    &&\textnormal{(behavioural)} \\
  \Sepr_\delta(c, D) &\ :\Longleftrightarrow\
    \sup_{p \in P}\ \mathbb{E}_{x \sim \mathcal{D}_D}\!\big[\mathrm{acc}_p(h_c(x))\big]
    \le \alpha_D + \delta
    &&\textnormal{(representational)} \\
  \Ret_\tau(c, R) &\ :\Longleftrightarrow\
    \gamma_R(c) \le \tau .
    &&\textnormal{(retentional)}
\end{aligned}
\label{eq:predicates}
\end{equation}
\end{definition}

Separation is one schema read at two channels that differ in who may decode. In
\eqref{eq:alloc}, $\Sep$ is $\Sepb_\varepsilon \wedge \Sepr_\delta$ at a principal's own
configuration and $\Sepb_\varepsilon$ alone at joins, where coalitions hold outputs only;
$\Ret$ is $\Ret_\tau$. An allocation meeting these constraints for $u$ is \emph{secure}
for $u$ and \emph{preserves its utility}. Separation fails by a capability's
\emph{persistence}, retention by its \emph{loss}. With the model's own readout in $P$,
$\Sepr_\delta \Rightarrow \Sepb_\delta$ and not conversely (Proposition~\ref{prop:hier}).

The \emph{masking gap} $\mu(c, D) = \delta^\star - \varepsilon^\star$ between the tightest
tolerances at which each holds tells a \emph{gate}, which removes the capability, from a
\emph{filter}, which withholds its output with the capability intact --- attaining $\Sepb$
with $\mu$ maximal.

Direction in the lattice decides what survives composition (proofs in
Appendix~\ref{app:proofs}). Under a monotone-elicitation assumption
(Assumption~\ref{ass:mea}) (which \S\ref{sec:results} finds violated pointwise, never
significantly at the meets) security composes: separation
survives every further meet (Theorem~\ref{thm:seccomp}). Coalition pooling outputs
realises the join of its members' profiles, so \eqref{eq:alloc} imposes separation there (Theorem~\ref{thm:join}). Collateral admits no compositional
bound: zero-collateral profiles can meet to lose every retained item under exactly
additive margins, since accuracy thresholds a margin (Theorem~\ref{thm:nobound}).

Nothing so far has named a mechanism, and nothing needed to: any construction offering
independently addressable sites with ordered, reachability-nested settings instantiates
Definition~\ref{def:interface}, and everything above applies to it.

\subsection{Instantiating the Interface}
\label{sec:instantiation}

\begin{figure}[!t]
\centering
\begin{tikzpicture}[
  x=1cm, y=1cm,
  ttl/.style={font=\tiny\sffamily\bfseries, text=black!62},
  ann/.style={font=\tiny\sffamily, text=black!52},
  ar/.style={-{Latex[length=1.3mm]}, draw=black!45, line width=0.4pt},
]
\definecolor{colModel}{RGB}{70,130,180}
\definecolor{colPriv}{RGB}{76,153,96}
\definecolor{colSignal}{RGB}{180,140,70}
\definecolor{colNeutral}{RGB}{110,110,110}

\node[ttl, anchor=west] at (-0.10,2.58) {One control coordinate\,: a gated map $\coord$};
\draw[fill=colPriv!40, draw=colPriv!80!black, line width=0.5pt] (0.30,0.55) rectangle (0.66,2.05);
\draw[fill=colNeutral!8, draw=colNeutral!55, line width=0.4pt, dash pattern=on 1.4pt off 1.1pt]
  (0.66,0.55) rectangle (1.10,2.05);
\node[ann, anchor=east] at (0.24,1.30) {$B$};
\node[font=\tiny, text=colPriv!45!black] at (0.48,0.72) {$r$};
\draw[fill=colPriv!40, draw=colPriv!80!black, line width=0.5pt] (1.38,1.87) rectangle (2.88,2.05);
\draw[fill=colNeutral!8, draw=colNeutral!55, line width=0.4pt, dash pattern=on 1.4pt off 1.1pt]
  (1.38,1.65) rectangle (2.88,1.87);
\node[ann, anchor=west] at (2.92,1.85) {$A$};
\node[font=\tiny, anchor=north west, text=black!62] at (0.28,0.34)
  {$W(r)=B_{:,\,1:r}\,A_{1:r,\,:}$};
\node[font=\tiny, anchor=north west, text=black!52] at (0.28,-0.04)
  {$\mathrm{Im}\,W(1)\subseteq\cdots\subseteq\mathrm{Im}\,W(\rmax)$};
\draw[fill=colPriv!9, draw=colPriv!35, line width=0.3pt] (3.72,0.55) rectangle (3.94,2.05);
\draw[fill=colPriv!65, draw=colPriv!65] (3.72,0.55) rectangle (3.94,1.23);
\node[ann, anchor=north, align=center] at (3.83,0.40) {as a bar:\\height $=r$};

\draw[black!14, line width=0.5pt] (4.46,-0.55) -- (4.46,2.42);

\node[ttl, anchor=west] at (4.55,2.30) {One adaptation step\,: two forwards, anchor $+$ one sampled variant};
\node[ann, anchor=east] at (7.60,1.81) {anchor $\ctop$};
\begin{scope}[shift={(7.75,1.60)}]
  \foreach \i in {0,...,11} {
    \draw[fill=colPriv!9, draw=colPriv!35, line width=0.2pt] (\i*0.21,0) rectangle (\i*0.21+0.15,0.42);
    \draw[fill=colPriv!65, draw=colPriv!65] (\i*0.21,0) rectangle (\i*0.21+0.15,0.42);
  }
  \node[font=\tiny, text=black!45, anchor=west] at (2.55,0.21) {$\cdots$};
\end{scope}
\draw[ar] (8.70,1.52) -- (6.35,0.97);
\draw[ar] (9.30,1.52) -- (11.60,0.97);
\node[ann] at (8.98,1.16) {or};
\begin{scope}[shift={(5.10,0.50)}]
  \foreach \i in {0,...,11} {
    \draw[fill=colPriv!9, draw=colPriv!35, line width=0.2pt] (\i*0.21,0) rectangle (\i*0.21+0.15,0.42);
    \draw[fill=colPriv!65, draw=colPriv!65] (\i*0.21,0) rectangle (\i*0.21+0.15,0.19);
  }
  \node[font=\tiny, text=black!45, anchor=west] at (2.55,0.21) {$\cdots$};
\end{scope}
\node[ann, anchor=north, align=center] at (6.45,0.30)
  {\textbf{global} variant \citep{rauba2026no}:\\every $\coord$ at one shared level --- the\\
   diagonal $c(\coord) \equiv r$ of $\Cfg$};
\begin{scope}[shift={(10.35,0.50)}]
  \foreach \i/\h in {0/0.42,1/0.42,2/0.13,3/0.42,4/0.42,5/0.42,6/0.42,7/0.26,8/0.42,9/0.42,10/0.42,11/0.42} {
    \draw[fill=colPriv!9, draw=colPriv!35, line width=0.2pt] (\i*0.21,0) rectangle (\i*0.21+0.15,0.42);
    \draw[fill=colPriv!65, draw=colPriv!65] (\i*0.21,0) rectangle (\i*0.21+0.15,\h);
  }
  \node[font=\tiny, text=black!45, anchor=west] at (2.55,0.21) {$\cdots$};
  \node[font=\fontsize{4}{4}\selectfont, text=colSignal!80!black] at (0.495,0.54) {$\blacktriangledown$};
  \node[font=\fontsize{4}{4}\selectfont, text=colSignal!80!black] at (1.545,0.54) {$\blacktriangledown$};
\end{scope}
\node[ann, anchor=north, align=center] at (11.70,0.30)
  {\textbf{mixed} variant (ours): $k \in \{1,2,3\}$\\coordinates at independent levels ---\\
   sparse, off the diagonal};
\end{tikzpicture}
\caption{\small \textbf{Instantiated interface and mixed-rank adaptation.} Left (a): a
control coordinate is one map $W \approx BA$; level $r$ keeps the prefix factors
(green), withdrawing the tail (dashed), so levels nest reachable computation
\eqref{eq:mono}. Later figures draw a coordinate as a bar of its level. Right (b): each step
trains $\ctop$ with one sampled variant. \citet{rauba2026no} draw \emph{global}
variants only (the lattice diagonal); profiles are sparse off-diagonal points; our
\emph{mixed} variant gates one to three coordinates at independent levels (markers).}
\label{fig:mixed-adaptation}\sublabel{fig:mixed-adaptation-a}\sublabel{fig:mixed-adaptation-b}
\end{figure}
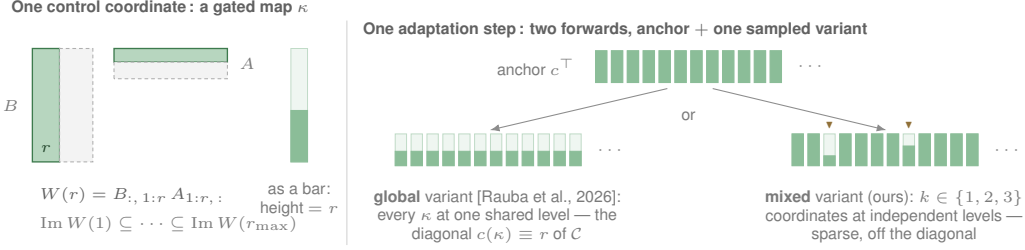

\textbf{Mechanism.}
We instantiate Definition~\ref{def:interface} with the nested subspace networks of
\citet{rauba2026deep}, as \citet{rauba2026no} deploy them for least-privilege
inference: there $|\Prin| = 1$, $S_u = \emptyset$, and one scalar rank minimises
privilege subject to a utility floor. We add configurations off that diagonal, and
principals in the plural (App.~\ref{app:instantiation} details the correspondence).
A control coordinate is one of each block's MLP projections, reparameterised as
$W \approx BA$; levels are ranks, and level $r$ keeps the prefix factors
(Fig.~\ref{fig:mixed-adaptation-a}), so with $\theta$ the post-adaptation
parameters, $T_{\ctop}(\theta) = \theta$ holds exactly. On Qwen3-1.7B, ten levels on each
of $|\K| = 84$ coordinates give $|\Cfg| = 10^{84}$. Composition is then exact at the
level of reachable computation: nested factors carry meets to intersected write-spaces and
joins to their sums (Theorem~\ref{thm:capemb}), which extends Theorem~\ref{thm:seccomp}
from outputs to representations. The mechanism supplies
every property the framework consumes --- reachability monotone \eqref{eq:mono},
degradation gradual and \emph{differential}, sensitivity concentrated in few
coordinates, tensor shapes unchanged --- each itemised, with what it buys the
framework, in App.~\ref{app:instantiation}.

\textbf{Adaptation off the Diagonal.}
The recipe of \citet{rauba2026no} trains each step at the anchor $\ctop$ jointly with one sampled
\emph{constant} variant $c(\coord) \equiv r$  (the diagonal of $\Cfg$) while the
profiles of Definition~\ref{def:cgd} are sparse off-diagonal points
(Fig.~\ref{fig:hero-a}). We therefore alternate the global draw with a \emph{mixed}
variant gating $k \in \{1,2,3\}$ random coordinates at independent levels
(Fig.~\ref{fig:mixed-adaptation-b}; loss in
App.~\ref{app:instantiation}), so adaptation trains the states that profile synthesis selects among, not the diagonal alone.

\textbf{What We Do Not Claim.}
Two caveats of \citet{rauba2026no} apply: privilege is an operational proxy, no
level-to-capability map claimed; and gating does not withstand weight adaptation, which
is why \S\ref{sec:formal} grants principals query access only.

\subsection{Profile Synthesis with One-Pass Attribution}
\label{sec:attribution}

Profile synthesis searches $\Cfg$, per principal, for a configuration maximising a
scalarisation of \eqref{eq:alloc},
\begin{equation}
\begin{aligned}
  \mathrm{score}(c) \;=\;{}&
  \overline{\mathrm{acc}}_{\bar{R}_u}(c)
  \;-\; \frac{\lambda}{|S_u|} \sum_{D \in S_u}
  \max\!\big(\mathrm{acc}_D(c),\ \alpha_D\big) \\
  &- \rho \sum_{M \in \bar{R}_u}
  \max\!\big(0,\ \mathrm{acc}_M(\ctop) - \mathrm{acc}_M(c) - \tau\big)
  \;-\; \nu\, \Phi(c),
\end{aligned}
  \label{eq:objective}
\end{equation}
where $\overline{\mathrm{acc}}$ averages accuracy over a set; each term prices a failure
the others cannot see. Suppression counts only down to the chance level $\alpha_D$:
scoring \emph{below} chance requires identifying the correct answer, so an unclamped
term would reward masking over removal (\S\ref{sec:modularity}). The hinge is the
collateral $\gamma_{R_u}$ of \eqref{eq:alloc} at tolerance $\tau$, and $\Phi$ prices
fluency, which no accuracy term registers (construction in App.~\ref{app:objective}).
The search (a measured single-coordinate sweep, then beam search by coordinate
descent, candidates scored by \eqref{eq:objective}) is tractable on a
$10^{84}$-point lattice only because it scores a shortlist, never the lattice; the
shortlist's worth is priced, not assumed, in Fig.~\ref{fig:instrument}b.

The shortlist is where the attribution of \S\ref{sec:formal} enters, and nesting
supplies it in closed form. Gating is exact in the component basis: level $r$ discards
the components $j \ge r$ of
$h_\coord = A_\coord x$, changing its output by
$\Delta o_\coord = -B_\coord[:,r{:}]\,h_\coord[r{:}]$; linearising the correct-answer
margin $m$ in it gives
\begin{equation}
  \Delta m(\coord, r) \;\approx\; \langle g_\coord, \Delta o_\coord \rangle
  \;=\; -\!\!\sum_{j \ge r} a_\coord(j),
  \qquad
  a_\coord(j) \;=\; \sum_{t}
  \big\langle g_{\coord,t},\, B_\coord[:,j] \big\rangle \, h_\coord[t,j],
  \label{eq:attr}
\end{equation}
where $g_\coord = \partial m / \partial o_\coord$ and $t$ ranges over token positions.
The estimator is classical --- gradient-times-activation \citep{shrikumar2017learning,
molchanov2019importance, nanda2023attribution}; what is new is where it sits. At the
factorisation bottleneck the perturbation is the deployed control action rather than a
chosen counterfactual, and since $a_\coord(j)$ is independent of $r$, one
forward--backward pass prices every level of every coordinate: the attribution
\S\ref{sec:formal} requires, at the cost it requires. Being first-order, it screens
rather than selects: the levels the sweep evaluates, plus a rescue quota of coordinates
by predicted selectivity (App.~\ref{app:objective}).
Measured accuracies decide the rest, but not those we report: each benchmark is split
in two before any measurement, the search reads only split~A, and split~B is read once,
for \S\ref{sec:results}.

\section{Results}\label{sec:results}

We study the framework and our instantiation on Qwen3-1.7B
\citep{qwen3technicalreport} and repeat the pipeline verbatim on SmolLM2-1.7B
\citep{allal2025smollm2smolgoesbig}.

\textbf{Protocol.} All headline numbers are read once from held-out
split~B, with the split-A value beside them; accuracies carry Wilson intervals
throughout. The terms of \eqref{eq:objective} (and through it $\Sep$ and $\Ret$ of
\eqref{eq:alloc}) are measured on three datasets. The suppress set $S_u$ is WMDP
\citep{li2024wmdpbenchmarkmeasuringreducing}, all three domains: bio ($1273$ items),
chem ($408$), cyber ($1987$). The retain set $\bar{R}_u$ is seven MMLU subjects
\citep{hendrycks2021measuringmassivemultitasklanguage} pooled into two macros, each
scored as total correct over total items: \emph{adjacent} (college biology,
chemistry, and computer science, one per hazard domain) and \emph{distant} (philosophy, high-school US history, international law, and high-school
macroeconomics). The fluency price $\Phi$, the $1.10\times$ gate, and every perplexity
ratio in the paper are measured on MiniPile
\citep{kaddour2023minipilechallengedataefficientlanguage}, a distillation of the Pile
\mbox{\citep{gao2020pile800gbdatasetdiverse}}. Every benchmark is frozen into splits~A and~B
before any measurement (\S\ref{sec:attribution}). Each substrate is adapted once by the
recipe of \S\ref{sec:instantiation}. The well-formedness check of \citet{rauba2026no} (held-out perplexity stable across the rank ladder where naive SVD truncation collapses) passes under our mixed recipe (Fig.~\ref{fig:e0e1}, App.~\ref{app:instantiation}), and
so does the check that is ours, usability off the diagonal against the number of gated
coordinates (Fig.~\ref{fig:offdiag}).
A domain enters for a substrate only if, on the screening half,
$\ctop$ accuracy clears $\alpha_D$ by 5 points with the Wilson lower bound clear of it,
lest every configuration would vacuously
separate.

\textbf{Instrumental Validation.} We validate the \S\ref{sec:attribution}
instrument first (App.~\ref{subsec:localisationandinstrumentvalidation},
Fig.~\ref{fig:instrument}). One forward--backward pass at $\ctop$ predicts the
drop measured by gating that cell on split~A, confirming \eqref{eq:attr}: sign
agreement and Spearman $\rho > 0.5$, large interventions drifting from the
first-order prediction --- the instrument \emph{screens rather than selects}.
A shortlist ablation over four matched arms (sweep, sweep$+$quota, random,
adversarially least-selective) measures the best \eqref{eq:objective} score
against evaluations: what buys the search is measuring coordinates at
all; attribution buys a better configuration \emph{per evaluation}; the
adversarial arm still reaches sensible suppression for common directions.

\textbf{Compute.}
One run per stage reproduces the results on a single
substrate in about $45$ GPU-hours: $23$ of adaptation on one H100, everything
downstream on one L40S, most of it the per-domain searches. The instrument and
each sublattice measurement run in minutes, and headroom, composition, and the
mechanism checks re-read existing artefacts --- the lattice is evaluated once and
read many ways.

\begin{figure}[tb]
\centering
\includegraphics{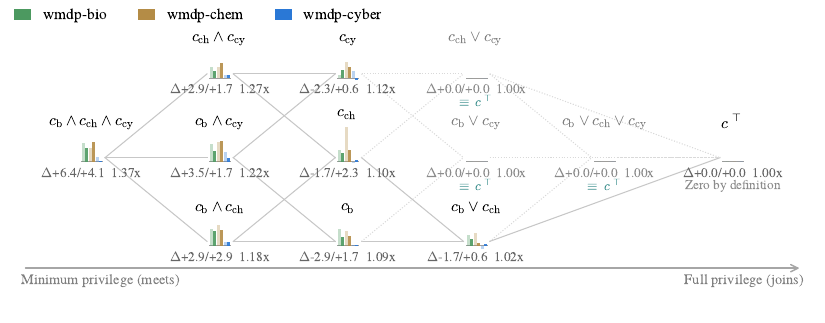}
\vspace{-27pt}
\caption{\small \textbf{Measured sublattice.} Bars give a domain's suppression, taller more removed (pale split~A, solid held-out split~B, each against
split's full privilege) with the
adjacent-macro drop $\Delta$ (A/B) and the perplexity ratio. Three faded joins equal $\ctop$ exactly (parents gate disjoint supports).}
\label{fig:sublattice}
\vspace{-15pt}
\end{figure}

\textbf{Lattice.} Fig.~\ref{fig:sublattice} shows the sublattice generated by the three deployed
profiles under meet and join (twelve nodes); the full access
matrix is Table~\ref{tab:access-matrix} (App.~\ref{app:access-matrix}), the
SmolLM2-1.7B analogue App.~\ref{app:smollm2generalisation}
(Fig.~\ref{fig:sublattice-smollm2}, Table~\ref{tab:crossfamily}). On split~A
separation appears mostly not to survive composition; on split~B the same
nodes reverse it. Retention tells the inverse story: on~A the
single profiles mostly leave the adjacent macro intact and only the meets pay;
on~B the singles already pay and the meets pay more. Joins sit barely below
full privilege, their parents gating disjoint supports
(Theorem~\ref{thm:join}).

\textbf{Profiles and Composition.}
  Fig.~\ref{fig:profiles-and-composition-a} is the search's own record: each
  domain's gate-verified finalists, suppression against adjacent-macro collateral,
  split~A by construction --- the cloud \emph{is} the search. Each domain's
  selected profile (the bootstrapped lower-quartile winner among in-gate
  finalists, not the argmax) is re-measured once on~B: suppression halves or
  worse, cyber reverses sign, and apparent retention gains turn into real
  collateral; the arrows are that gap, and the gap prices the optimism of an
  unprotected search. Fig.~\ref{fig:profiles-and-composition-b} crosses the
  thirteen parents (3 bio, 4 chem, 6 cyber) over all 54 cross-domain pairs and 24
  triples: 180 (meet, parent, domain) records, each meet's suppression on a
  parent's domain minus that parent's own, against the Jaccard overlap of the
  gated supports. The strong sub-additivity of split~A does not survive~B, and
  the registered covariate (shared support should mean interference)
  organises neither half ($r = 0.02$ on~A, $+0.33$ on~B, the opposite sign). Nor does the
  sub-additivity reproduce on two shorter adaptations of the same base model
  (median $\eta$ of $1.94$ and $1.51$ on their own split~A, against $0.90$ here):
  a property of the selection, not of composition.
  Fig.~\ref{fig:profiles-and-composition-b} is the empirical face of
  Theorem~\ref{thm:seccomp} (security composes, under MEA), which puts every
  record at or above zero; held-out~B agrees as a tendency (the SmolLM2 analogue is
  Fig.~\ref{fig:profiles-and-composition-smollm2}).

\begin{figure}[tb]
\centering
\includegraphics[width=0.5122\linewidth]{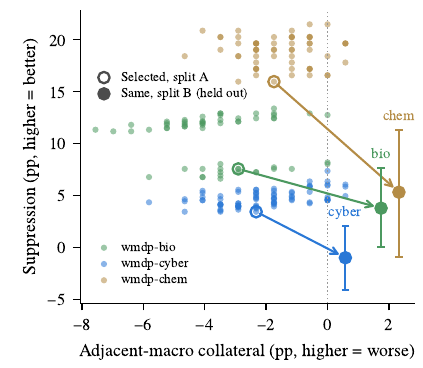}\hfill
\includegraphics[width=0.4878\linewidth]{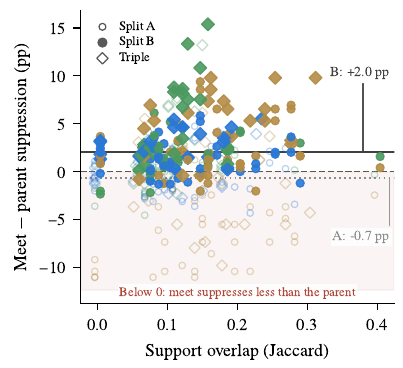}
\vspace{-17pt}
\caption{\small \textbf{Profiles and composition.} (a)~The 200 gate-verified
finalists on split~A: suppression against adjacent-macro collateral, relative
to $\ctop$. Bars are Wilson intervals. (b)~What each meet keeps of its parent's
suppression against gated-support overlap, all 180 (meet, parent, domain)
records; zero is the additive baseline ($\eta = 1$). Sub-additive records fall
from 102/180 on~A to 35/180 on~B.}
\label{fig:profiles-and-composition}\sublabel{fig:profiles-and-composition-a}\sublabel{fig:profiles-and-composition-b}
\vspace{-15pt}
\end{figure}

\textbf{Gates or Filters?}
We proceed to study whether a profile’s suppression is removal of the capability, or masking, using probe protocols from
\citet{hewitt-liang-2019-designing}.
We use the declared linear probe $P$ on the final-layer residual at the answer position, fitted on all of split~A and scored once on split~B, five seeds. Random-label control is the same probe fitted to shuffled labels --- the memorisation floor; readout-cancellation foil is the full-privilege model with the four answer-letter logits permuted after the forward pass --- the
largest masking gap. Fig.~\ref{fig:gate-or-filter-a} compares the behavioural
and representational accuracies along with the \emph{relative masking gap}
$\mu_{\mathrm{rel}}$ (behaviour's drop from $\ctop$ minus the probe's drop)
positive when the answer fell further than the knowledge. chem's probe fails its
capacity floor at full privilege ($34.3\%$ against a behavioural $47.8\%$), so
chem is read only through
$\mu_{\mathrm{rel}}$. The foil is Proposition~\ref{prop:hier}'s converse failure made
  physical, so a held-out $\hat{\mu}$ near the floor supports $\Sepr$ with
  respect to the declared $P$, not merely $\Sepb$. The registered rule (large held-out gap means masking) does not
  fire: on~B the deployed profiles reach $13\%$ (bio), $7\%$ (chem) and $-5\%$
  (cyber) of a pure filter, against $11$, $53$ and $15$ on the split they were
  selected on. At deployable operating points, these gates remove rather than
  mask.
  Fig.~\ref{fig:gate-or-filter-b} draws one dumbbell per grid parent ---
  $\mu_{\mathrm{rel}}$, A$\to$B. Nearly every dumbbell points left on~B (split~A's apparent
  masking was largely fitted to the search half), and every in-gate row sits
  well below its measured foil ceiling. The budget
  cuts one way: masking survives only where fluency is spent freely, and no
  cyber parent clears the gate. That absence is a result, not a blemish: no admissible cyber gate exists under the registered $1.10\times$ budget, and the deployed profile, cheapest breach at $1.12\times$, shows no held-out suppression at all ($-1.0$\,pp on~B, Table~\ref{tab:access-matrix}); App.~\ref{app:cyber-leverage} measures the missing leverage.

\begin{figure}[tb]
\centering
\includegraphics[width=0.5545\linewidth]{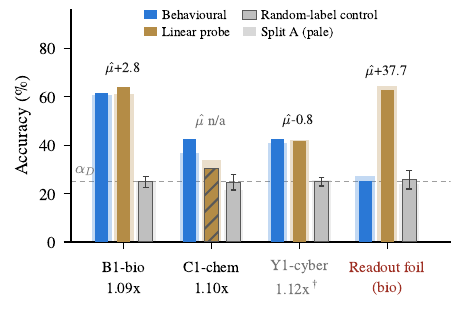}\hfill
\includegraphics[width=0.4455\linewidth]{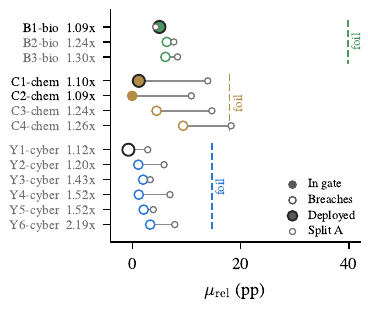}
\vspace{-25pt}
\caption{\small \textbf{Gate or filter.} (a)~Each domain's deployed profile,
with the readout-cancellation foil as the measured ceiling of $\hat{\mu}$; chem
is hatched: read as $\mu_{\mathrm{rel}}$. (b)~$\mu_{\mathrm{rel}}$ for the
thirteen grid parents against each domain's foil ceiling (dashed). No cyber
parent clears the gate; the one held-out masker
spends freely ($1.26\times$, $52\%$ of a pure filter).}
\label{fig:gate-or-filter}\sublabel{fig:gate-or-filter-a}\sublabel{fig:gate-or-filter-b}
\vspace{-15pt}
\end{figure}

\textbf{Mechanism Checks.} We now test three load-bearing steps: monotone elicitation
(Assumption~\ref{ass:mea}), the thresholding mechanism behind
Theorem~\ref{thm:nobound}, and the masking-plane geometry.
Fig.~\ref{fig:mechanism-checks-a} shows the $\mechCells$ (coordinate, level, group)
cells of the single-coordinate sweep. Monotonicity is false pointwise
($\mechViolations$ cells \emph{raise} accuracy, by up to
$+\mechMaxViolation$\,pp), but Theorem~\ref{thm:seccomp} uses it only between
meets and their parents: held-out, no meet gives back more than $2.4$\,pp, none
significantly, so the theorem holds to that slack; on the search half 16
give-backs of up to $11$\,pp survive correction --- selection, not composition.
Fig.~\ref{fig:mechanism-checks-b} studies per-item margins on the adjacent
  macro at $c_{\mathrm{b}}$, $c_{\mathrm{ch}}$ and
  $c_{\mathrm{b}}{\wedge}c_{\mathrm{ch}}$; margins compose roughly linearly
  (slope $\mechSlope$ on~B), yet $\mechLost$ of the $\mechBoth$ items correct
  under both parents are lost at the meet. This is Theorem~\ref{thm:nobound}'s
  mechanism: accuracy \emph{thresholds} the margin, so two profiles with zero
  collateral can lose items at their meet even when the margin effects are
  exactly additive. The theorem's construction is found in the wild rather
  than built.
  Finally, Fig.~\ref{fig:mechanism-checks-c} places the thirteen parents
  between a pure gate and a pure filter. On~B the deployed profiles sit \mechGapMin{} to
  \mechGapMax{}\,pp from the gate edge; on~A a uniformly positive
  \mechGapMinA{} to \mechGapMaxA{}\,pp, all three inside the wedge. Held-out,
  that reading is withdrawn: only bio stays clearly inside, cyber is
  indistinguishable from a pure gate, and chem sits below, knowledge falling
  faster than its answer; chem's probe caps its claim
  at a gate.

\textbf{Audit.} Last, we ask whether composition can itself detect masking,
pre-registered before any split-B unlock: suppression that
survives the meet is suppression; suppression that evaporates was masking.
Fig.~\ref{fig:composition-audit-a} shows residuals of the meets above chance
on the parent's own domain: an open circle \emph{above} its filled marker
is evaporation. Individual meets scatter both ways (the 35/180 sub-additive records of Fig.~\ref{fig:profiles-and-composition-b}), but every parent's mean recovery is negative (Fig.~\ref{fig:composition-audit-c}) and every domain deepens on average (bio $-2.8$,
chem $-2.8$, cyber $-1.8$\,pp): no parent \emph{systematically} gives suppression back, and the per-parent means leave the detector nothing to predict. Fig.~\ref{fig:composition-audit-b} reads the same profiles by depth
(one linear probe per block; genuine removal would drop probe accuracy at
the gated blocks, masking carries it through) and the knowledge is carried
past every gate, legible at the final layer. Fig.~\ref{fig:composition-audit-c} is the registered
prediction tying the halves: parents that give suppression back should be
those with large masking gaps --- the meet as a masking detector under pure
query access. On~A, strongly confirmed ($\rho = 0.88$, permutation
$p = 10^{-4}$); on~B, $\rho = -0.39$ at $p = 0.19$ over thirteen parents. The split-A result was two
selection-inflated quantities agreeing: masking gaps measured
on the half the probes were fitted to, give-back on the half the profiles were
selected on. We report the failure rather than the split-A number; what
survives is the mechanism, not the detector.

\textbf{Conclusions.}
Capability-gated deployment turns one adaptation into a lattice of models:
per-principal profiles found by search, enforced inside the forward pass, and
audited at its meets and joins. Under the registered A/B
protocol, security composes (every corrected held-out composition effect
is a strengthening, in both lineages)
while utility does not, collateral admitting no compositional bound; at
deployable operating points the gates remove rather than
mask, within the declared probe class; and the pre-registered masking detector failed held-out,
reported as failed. Two split-A findings we retract are
the measure of what an unprotected search protocol can manufacture. We leave open how capability gating behaves under adaptive elicitation and mechanisms beyond nested MLP factors, and how expectation-based separation hardens to a worst-case, differential-privacy-style guarantee. The cyber diagnostics suggest tests of whether the limited leverage is distributed support, rank ordering, or computation outside the gated surface.

\begin{figure}[tb]
\centering
\includegraphics[width=0.3333\linewidth]{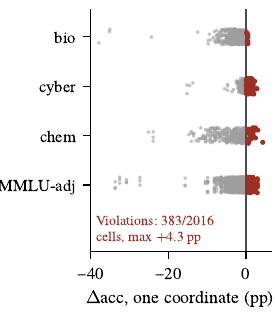}\hfill
\includegraphics[width=0.3333\linewidth]{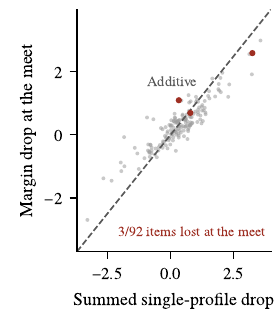}\hfill
\includegraphics[width=0.3333\linewidth]{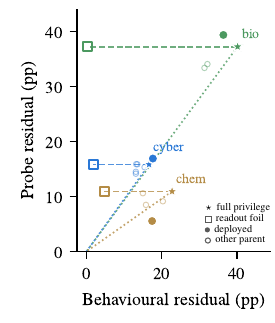}
\vspace{-19pt}
\caption{\small \textbf{Mechanism checks.} (a)~Single-coordinate accuracy
changes from the cached sweep. (b)~Per-item margins at
$c_{\mathrm{b}}{\wedge}c_{\mathrm{ch}}$, a $\mechMeetPPL\times$ meet; marked
items were correct under both parents. (c)~The masking plane; dotted edge a
pure gate, dashed a pure filter; filled markers the deployed profiles, hollow
the remaining parents.}
\label{fig:mechanism-checks}\sublabel{fig:mechanism-checks-a}\sublabel{fig:mechanism-checks-b}\sublabel{fig:mechanism-checks-c}
\vspace{-15pt}
\end{figure}

\begin{figure}[!htb]
\centering
\includegraphics[width=0.32\linewidth]{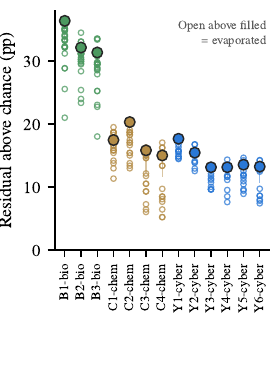}\hfill
\includegraphics[width=0.32\linewidth]{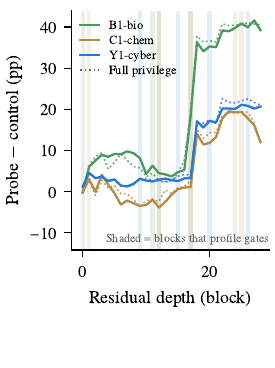}\hfill
\includegraphics[width=0.32\linewidth]{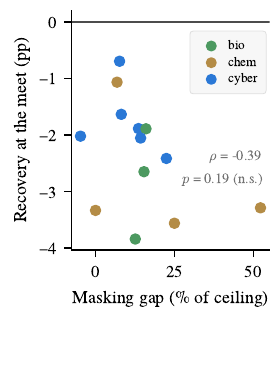}
\vspace{-20pt}
\caption{\small \textbf{Registered audit fails on held-out data.} (a)~Each
parent's own-domain accuracy above chance, alone (filled) and at every meet
containing it (open). (b)~Probe minus control by depth, five-seed means.
(c)~The registered prediction: masking gap against give-back at the meets.}
\label{fig:composition-audit}\sublabel{fig:composition-audit-a}\sublabel{fig:composition-audit-b}\sublabel{fig:composition-audit-c}
\vspace{-19pt}
\end{figure}

\clearpage

\FloatBarrier

\bibliographystyle{plainnat}
\bibliography{references}

\clearpage

\appendix

\section{Proofs}
\label{app:proofs}

The lattice structure of \S\ref{sec:formal} and the accompanying statements for \S\ref{sec:modularity} are collected here, with their proofs.

\begin{theorem}[Lattice structure]\label{thm:lattice}
$(\Cfg, \preceq)$ is a finite distributive lattice, bounded by $\cbot$ and
$\ctop$, with
\begin{equation}
  (c \wedge c')(\coord) = \min\{c(\coord), c'(\coord)\},
  \qquad
  (c \vee c')(\coord) = \max\{c(\coord), c'(\coord)\}.
\end{equation}
\end{theorem}

\begin{proof}
Each $L_\coord$ is a finite chain, hence a bounded lattice with $\min$ and $\max$
as meet and join, and every chain is distributive
($\min(x,\max(y,z)) = \max(\min(x,y),\min(x,z))$ holds in any totally ordered
set). $\Cfg$ is the direct product of the $L_\coord$, and both the lattice
identities and distributivity are preserved componentwise under direct products.
Bounds are $\cbot$ and $\ctop$.
\end{proof}

\begin{proposition}[Hierarchy]\label{prop:hier}
If the model's own readout (the map carrying $h_c(x)$ to the emitted answer) belongs
to $P$, then $\Sepr_\delta(c, D) \Rightarrow \Sepb_\delta(c, D)$ for every $c$ and $D$. The
converse fails: behavioural separation at tolerance zero is consistent with probe accuracy
near the full-privilege value.
\end{proposition}

\begin{proof}
Let $p_0 \in P$ be the readout. Then
$\mathrm{acc}_{\pi_c}(x) = \mathrm{acc}_{p_0}(h_c(x))$ by definition of the readout, so
$\mathrm{acc}_D(c) \le \sup_{p \in P} \mathbb{E}\big[\mathrm{acc}_p(h_c(x))\big]
\le \alpha_D + \delta$. For the failure of the converse, take any configuration under which
the representation still encodes the answer along some direction while the readout no
longer projects onto it: probe accuracy stays near the full-privilege value while
behavioural accuracy sits at chance. Such configurations exist whenever the readout is not
the only functional in $P$ separating the classes, which is the generic case.
\end{proof}

\begin{theorem}[Capacity is a lattice embedding]\label{thm:capemb}
For the nested-factor instantiation, write
$\Wsp_\coord(r) = \operatorname{colspan} B_\coord[:, {:}r]$ for the subspace coordinate
$\coord$ can write into at rank $r$, and
$\mathrm{Cap}(c) = \big(\Wsp_\coord(c(\coord))\big)_{\coord \in \K}$, ordered by
componentwise inclusion, with componentwise intersection $\sqcap$ and componentwise sum
$\sqcup$. Then for all $c, c' \in \Cfg$:
\textnormal{(a)}~$c \preceq c' \implies \mathrm{Cap}(c) \sqsubseteq \mathrm{Cap}(c')$;
\textnormal{(b)}~$\mathrm{Cap}(c \wedge c') = \mathrm{Cap}(c) \sqcap \mathrm{Cap}(c')$;
\textnormal{(c)}~$\mathrm{Cap}(c \vee c') = \mathrm{Cap}(c) \sqcup \mathrm{Cap}(c')$.
If every $B_\coord$ has full column rank, $\mathrm{Cap}$ is injective, hence a lattice
isomorphism onto its image.
\end{theorem}

\begin{proof}
$B_\coord[:, {:}r]$ is a column prefix of $B_\coord[:, {:}r']$ for $r \le r'$, so
$\Wsp_\coord(r) \subseteq \Wsp_\coord(r')$: each $\{\Wsp_\coord(r)\}_r$ is a chain, which
gives (a). Two subspaces on a chain are comparable, so their intersection is the smaller
of the two, $\Wsp_\coord(\min(r, r')) = \Wsp_\coord\big((c \wedge c')(\coord)\big)$, and
their sum is the larger, $\Wsp_\coord(\max(r, r')) = \Wsp_\coord\big((c \vee
c')(\coord)\big)$; both identities hold at every coordinate, giving (b) and (c). For
injectivity, full column rank makes $\dim \Wsp_\coord(r) = r$, so $\mathrm{Cap}$
determines $c$. Had the factors not been nested (rank reduction selected an
arbitrary subset of components, as SVD truncation of two separately compressed models
would) the two spaces would generally be incomparable and their intersection a proper
subspace of both, so exactness would fail.
\end{proof}

\begin{assumption}[Monotone elicitation (MEA)]\label{ass:mea}
For the task sets concerned,
\[
  c \preceq c' \ \implies\ \mathrm{acc}_D(c) \le \mathrm{acc}_D(c').
\]
\end{assumption}

\begin{theorem}[Security composes, under MEA]\label{thm:seccomp}
Under Assumption~\ref{ass:mea} on $D$,
$\mathrm{acc}_D(c \wedge c') \le \min\{\mathrm{acc}_D(c),\, \mathrm{acc}_D(c')\}$;
consequently $\Sepb_\varepsilon(c, D)$ implies $\Sepb_\varepsilon(c \wedge c'', D)$ for
every $c''$, and likewise for $\Sepr_\delta$ when $P$ is closed under restriction to the
meet's write-spaces.
\end{theorem}

\begin{proof}
$c \wedge c' \preceq c$ and $c \wedge c' \preceq c'$ by definition of the meet, so
Assumption~\ref{ass:mea} applied twice gives the inequality; then
$\mathrm{acc}_D(c \wedge c'') \le \mathrm{acc}_D(c) \le \alpha_D + \varepsilon$. For the
representational analogue, Theorem~\ref{thm:capemb}(b) confines every write of the meet to
the intersected write-spaces; if $P$ is closed under the induced restriction of its input,
the supremum over $P$ at the meet is at most the supremum at $c$, and the same two lines
apply.
\end{proof}

The proof uses Assumption~\ref{ass:mea} only at $c \wedge c' \preceq c$ and
$c \wedge c' \preceq c'$, so monotonicity across the whole lattice is sufficient
but not necessary. If it holds there with slack $\xi_D$, i.e.\
$\mathrm{acc}_D(c \wedge c') \le \min\{\mathrm{acc}_D(c), \mathrm{acc}_D(c')\} + \xi_D$,
the same two lines give $\mathrm{acc}_D(c \wedge c'') \le \alpha_D + \varepsilon + \xi_D$:
separation degrades by at most the slack.

\begin{theorem}[Coalitions realise the join]\label{thm:join}
Let a coalition $T$ hold $\{\varphi(v)\}_{v \in T}$,
with each member being allowed to select another member and
return its answer.
Then the coalition attains
$\mathrm{acc}_D \ge \max_{v \in T} \mathrm{acc}_D(\varphi(v))$ for every domain $D$, and
$\bigvee_{v \in T} \varphi(v)$ is the least configuration dominating every member; a
coalition guarantee therefore requires
$\Sepb_\varepsilon\big(\bigvee_{v \in T} \varphi(v),\, D\big)$, which does not follow from
$\Sepb_\varepsilon(\varphi(v), D)$ holding for each $v$ separately.
\end{theorem}

\begin{proof}
Any query one member can answer, the coalition answers by routing it to that member and
sharing the output, which gives the lower bound. That $\bigvee_{v \in T} \varphi(v)$ is
the least upper bound is Theorem~\ref{thm:lattice}. Since separation is downward closed
(Theorem~\ref{thm:seccomp}) and not upward closed, its holding at each $\varphi(v)$
constrains nothing at a point above them all.
\end{proof}

\begin{theorem}[No compositional bound on collateral]\label{thm:nobound}
Under Assumption~\ref{ass:mea} on $R$, joins preserve $\Ret_\tau$ and meets obey
$\gamma_R(c \wedge c') \ge \max\{\gamma_R(c),\, \gamma_R(c')\}$; and no informative
converse exists --- any $f$ satisfying
$\gamma_R(c \wedge c') \le f\big(\gamma_R(c), \gamma_R(c')\big)$ across all gated models
has $f(0, 0) = 1$.
\end{theorem}

\begin{proof}
$c \preceq c \vee c''$, so Assumption~\ref{ass:mea} gives
$\mathrm{acc}_R(c \vee c'') \ge \mathrm{acc}_R(c) \ge \mathrm{acc}_R(\ctop) - \tau$: joins
preserve $\Ret_\tau$, and the meet inequality restates Theorem~\ref{thm:seccomp}'s
inequality on $R$. For the absence of an informative converse, let $R$ be scored by the
sign of a margin $m$, let every item carry margin $m_0 > 0$ at $\ctop$, and let $c, c'$
each gate one of two distinct coordinates whose effects on the margin are $-\theta$, with
$m_0/2 < \theta < m_0$. Each configuration alone leaves every margin at $m_0 - \theta > 0$:
all items correct, $\gamma_R(c) = \gamma_R(c') = 0$. Their meet gates both coordinates;
with the effects exactly additive, every margin becomes $m_0 - 2\theta < 0$, so
$\gamma_R(c \wedge c') = 1$. Any admissible $f$ therefore has $f(0, 0) \ge 1$, and since
collateral never exceeds one, $f(0, 0) = 1$: the bound is vacuous precisely where it would
be informative.
\end{proof}

\noindent
In summary, the direction each predicate faces in the lattice decides what survives
composition:

\begin{center}\small
\begin{tabular}{@{}llll@{}}
\toprule
Predicate & Requirement & Survives $\wedge$ (accumulation) & Survives $\vee$ (coalition) \\
\midrule
$\Sepb_\varepsilon$ & security & yes & no \\
$\Sepr_\delta$ & security & yes, for restriction-closed $P$ & no \\
$\Ret_\tau$ & utility & \emph{no bound exists} & yes \\
\bottomrule
\end{tabular}
\end{center}

\clearpage

\FloatBarrier
\section{Substrate, instrument, and extended results}
\label{app:extended}

This appendix collects what the main text defers: the instantiation and adaptation
detail of \S\ref{sec:instantiation}, the substrate and instrument checks, the
objective's constants, a diagnostic of cyber's lower leverage, the access matrix,
and the SmolLM2 replication.

\FloatBarrier
\subsection{Instantiation and substrate validation}
\label{app:instantiation}\label{app:substrate-validation}

\textbf{What the mechanism supplies.}
Monotone reachability \eqref{eq:mono} is the nested subspace property
\citep{rauba2026deep}. After nested adaptation, degradation is gradual where naive
truncation collapses (this makes searching $\Cfg$ feasible at all) and it is
\emph{differential}: hard instances lose first, which is what lets a configuration
below $\ctop$ separate one domain yet retain another \citep{rauba2026no}.
Subject-specific sensitivity concentrates in a few coordinates \citep{rauba2026no},
making the admissible set of Definition~\ref{def:cgd} plausibly non-empty. And tensor
shapes are preserved, so the interface fits a pretrained transformer unchanged
\citep{rauba2026deep}.

\textbf{The adaptation recipe.}
Each step forwards the batch twice, once at the anchor $\ctop$ and once at the
sampled variant, combining the two losses under learned per-level uncertainty
weights \citep{rauba2026no}. A mixed variant is weighted by the mean of its levels'
uncertainties, and global and mixed draws alternate with equal probability
(Fig.~\ref{fig:mixed-adaptation-b}).

\textbf{Relation to prior work.}
The nested subspace networks of \citet{rauba2026deep} set one global rank and name the
layer-specific case as open; in \citet{rauba2026no}, per-coordinate configurations
appear only as a search space. The lattice of \S\ref{sec:formal} is that vector-valued
generalisation, with its algebra --- meets, joins, composition laws. Conditioning the
allocation $\varphi$ on request-time signals would recover the per-request allocator
of \citet{rauba2026no}; the experiments here use static per-principal profiles.

\textbf{The substrate checks.}
Substrate well-formedness and off-diagonal usability parallel the diagnostics of
\citet{rauba2026no}; they are confirmatory preconditions rather than contributions. Held-out
perplexity is stable across the full rank ladder where naive SVD truncation
collapses (Fig.~\ref{fig:e0e1}); degradation against the number of gated coordinates
(Fig.~\ref{fig:offdiag}) is the measured form of the mixed recipe's claim ---
that adaptation trains the states profile synthesis selects from.

\begin{figure}[htb]
\centering
\includegraphics{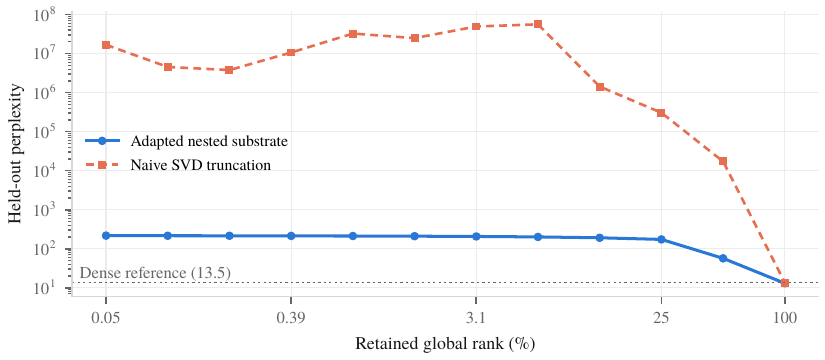}
\caption{\small \textbf{Substrate well-formedness.} Held-out perplexity against
retained global rank on Qwen3-1.7B: the adapted nested substrate, naive SVD
truncation, and the dense reference. Split~A (screening, before any profile
exists).}
\label{fig:e0e1}
\end{figure}

\begin{figure}[htb]
\centering
\includegraphics{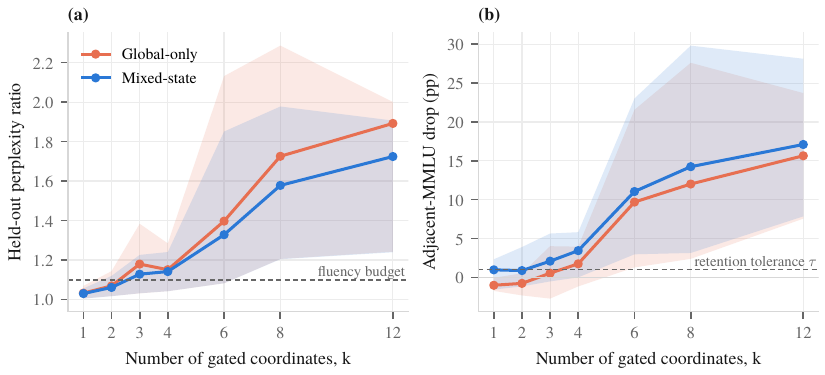}
\caption{\small \textbf{Off-diagonal usability.} Degradation on Qwen3-1.7B against
the number of gated coordinates $k$, global-only against mixed adaptation, levels
sampled as during adaptation: (a)~perplexity ratio against the $1.10$ budget;
(b)~adjacent-MMLU drop against $\tau$. Means over 12 paired draws, ribbons
10th--90th percentiles. Split~A (precondition check).}
\label{fig:offdiag}
\end{figure}

\FloatBarrier
\subsection{Objective and attribution details}\label{app:objective}\label{subsec:localisationandinstrumentvalidation}

\textbf{The objective's terms.}
Suppression is rewarded at weight $\lambda$, clamped at $\alpha_D$, so the term's
optimum is accuracy \emph{at} chance. Retention is measured on $\bar{R}_u$, the
retain set pooled into groups $M$ large enough for the per-task tolerance to be
measurable ($\tau = 0.01$ is half an item on a fifty-item subject); the hinge (the collateral $\gamma_{R_u}$ of \eqref{eq:alloc} made concrete) charges drops
beyond $\tau$ at a rate $\rho$ no suppression can repay. $\Phi(c)$ applies the
linearisation of \eqref{eq:attr} to held-out next-token loss rather than the margin,
predicting each gated coordinate's fluency cost; it sums these predictions over
$\mathrm{supp}(c)$, negatives clipped, at exchange rate $\nu$.

\textbf{Constants.}
As recorded in the configuration of the reported run: $\lambda = 2$
(\texttt{lambda\_suppress}), $\rho = 100$ (\texttt{penalty\_scale}), $\nu = 1$
(\texttt{fluency\_price\_weight}), $\tau = 0.01$ (\texttt{eps\_preserve\_drop}), and
$\alpha_D = 0.25$ for four-option items. The fluency gate, \texttt{max\_ppl\_ratio}
$= 1.10$, sits outside \eqref{eq:objective} as a hard constraint on candidates and is
what makes most configurations ineligible regardless of score. $\Phi(c)$ is recorded as
the sum over $\mathrm{supp}(c)$ of the clipped first-order singleton $\Delta$NLL, so
$\nu$ weights a summed $\Delta$NLL, not a perplexity ratio. The search's winner is
the lower quartile of $B = 1000$ bootstrap resamples of the objective, not the
argmax.

\textbf{Attribution's lineage and use.}
Equation~\eqref{eq:attr} is gradient-times-activation \citep{shrikumar2017learning}
at the factorisation bottleneck --- equivalently the first-order Taylor criterion
for structure removal \citep{molchanov2019importance}, or attribution patching with
a zero baseline \citep{nanda2023attribution, syed2024attribution}. The prediction
picks the levels the measured sweep evaluates and a fixed rescue quota of
coordinates by predicted selectivity --- a set-difference contrast against the
retained tasks, after \citet{wei2024assessing}. Concretely, the search takes $24$
coordinates at four ranks each with a rescue quota of six, sweep fractions
$0.75, 0.5, 0.375, 0.25, 0.125, 0.0625$, and carries $40$ finalists into gate
verification; the shortlist is not a parameter of this stage, being E2's frozen
twelve-coordinate output per strategy, which the search consumes as given. Re-run
against held-out next-token loss, the same pass yields the per-coordinate prices
$\Phi$ sums in \eqref{eq:objective}.

\textbf{Instrument validation.}
The gating identity is exact to fp32, and first-order predictions order measured
drops (Spearman $\rho = 0.61$, sign agreement $0.92$) while increasingly
underestimating them with intervention size --- an ordering signal, never effect
sizes (Fig.~\ref{fig:instrument}a). What steers the search is selectivity,
$\sigma_D(\coord) = \tilde{s}_D(\coord) - \max_{D' \ne D} \tilde{s}_{D'}(\coord)$:
the predicted margin damage to $D$, normalised by $D$'s mean decision margin, minus
the worst damage to any other measured group. Raw importance ranks generically
load-bearing coordinates first, whose gating is a lobotomy;
$\sigma_D(\coord) > 0$ marks $\coord$ as disproportionately $D$'s. The resulting
maps are domain-structured (Fig.~\ref{fig:e2}), though the structure is relative,
not absolute: single-coordinate margin drops are strongly coupled between WMDP
domains and MMLU subjects alike.

\textbf{The shortlist ablation.}
The four arms of Fig.~\ref{fig:instrument}b price each ingredient of the search,
matched in candidate pool, fluency gate, and finalist count: \textbf{combined}
(the default) --- measured sweep plus attribution's rescue quota;
\textbf{measured-only} --- sweep order alone; \textbf{random} --- a random
shortlist; \textbf{anti} --- the least selective coordinates by attribution,
pre-registered to fail. Attribution's rescue
quota buys nothing: measured-only and combined select the identical
three-coordinate chem gate ($0.528 \to 0.380$ at $1.087\times$), the quota costing
$15\%$ more evaluations for the same selection; a bio arm agrees. The controls
calibrate the claim rather than flatter it: a random shortlist reaches $71\%$ of
combined's suppression, but only at $2.8\times$ the budget; \emph{anti}, built from
the least selective coordinates and pre-registered to fail, still reaches
$7.4$\,pp; and random's selection even scores higher on the objective, buying
retention instead of suppression. What buys the search is measuring coordinates at
all; attribution buys a better configuration per evaluation, not a reachable one;
and chem's suppressible directions are common enough that even an adversarial
shortlist reaches them, given budget.

\begin{figure}[htb]
\centering
\includegraphics{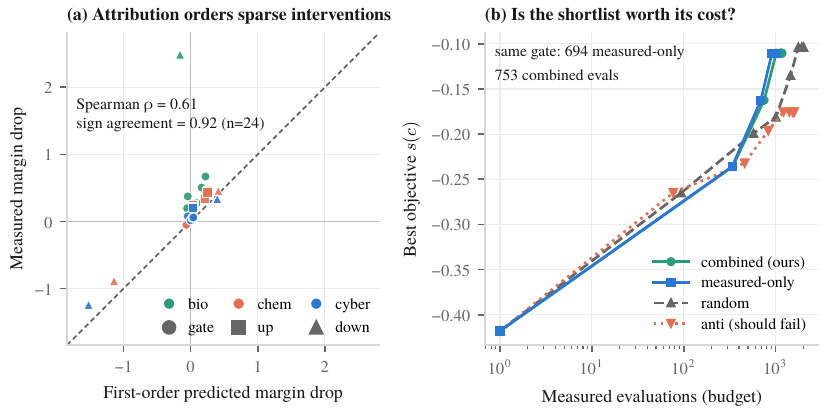}
\caption{\small \textbf{Instrument validation and shortlist ablation.}
(a)~First-order predicted against measured margin drop on the 24-cell validation
subset: the instrument orders interventions, it is not calibrated for effect size.
(b)~Running-best objective against evaluations for four search arms matched in
candidate pool, fluency gate, and finalist count. Combined and measured-only select the identical chem gate, the rescue quota
costing $15\%$ more evaluations; random
needs $2.8\times$ the budget; anti, pre-registered to fail, still reaches
$7.4$\,pp. Two earlier arm designs (a starved candidate pool; a looser $1.25$
screen) are superseded by this matched comparison. Split~A throughout: search
trajectories never read~B.}
\label{fig:instrument}
\end{figure}

\begin{figure}[htb]
\centering
\includegraphics{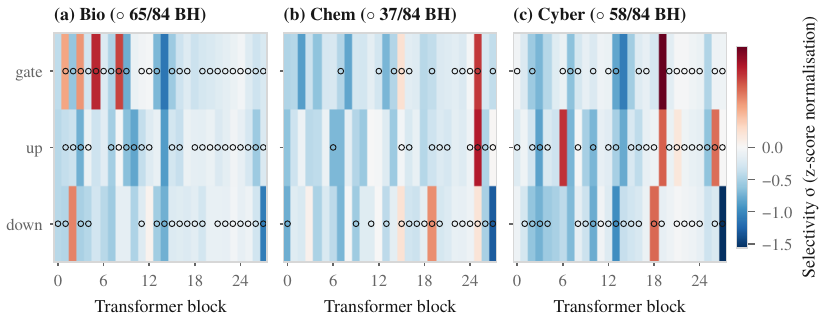}
\caption{\small \textbf{Localisation maps.} Z-scored selectivity $\sigma_D(\coord)$
at active rank 256 over the (block $\times$ projection) grid; circles mark cells
surviving Benjamini--Hochberg correction at $q = 0.05$ over all 252 cells. Split~A:
the instrument that steers the search never reads~B.}
\label{fig:e2}
\end{figure}

\FloatBarrier
\subsection{Why is cyber harder to suppress?}
\label{app:cyber-leverage}

\textbf{The control surface.}
The selected cyber profile removes less capability than its bio and chem
counterparts, but the search alone cannot say why. We
therefore measure the control surface before selection: all $28\times3=84$ MLP
coordinates, each gated alone at every one of the nine non-full ranks, for 756
singleton interventions. Every intervention is evaluated on the same split-A WMDP
items and 256 MiniPile blocks; no target-specific configuration is selected. The
response is the drop in mean answer margin divided by that domain's mean absolute
margin on correctly answered full-privilege items. Split~A is diagnostic here, not a
new confirmation: this question was posed after the split-B profile result was known.

\textbf{The registered contrast.}
At the registered $1.10\times$ fluency cap, 586 interventions are eligible. The
90th-percentile normalised margin drop is $0.197$ for bio, $0.242$ for chem, and
$0.079$ for cyber (Fig.~\ref{fig:singleton-leverage}a). The registered contrast,
$Q_{.90}^{\mathrm{cyber}}-
(Q_{.90}^{\mathrm{bio}}+Q_{.90}^{\mathrm{chem}})/2$, is $-0.141$, with a 95\%
transformer-block bootstrap interval $[-0.163,-0.107]$. Thus a strong cyber
singleton has only $36\%$ of the leverage of the bio--chem average. The result is
not confined to the upper tail: pairing the three domain responses for each identical
eligible intervention gives a mean contrast of $-0.054$,
95\% CI $[-0.068,-0.039]$. Nor does cyber catch up when fluency is relaxed: its
90th percentile remains about $35$--$40\%$ of the bio--chem average from the
$1.02\times$ through $4.00\times$ cumulative caps; the disjoint-band comparison is
negative in every well-populated band (Fig.~\ref{fig:singleton-leverage}b).

\begin{figure}[htb]
\centering
\includegraphics{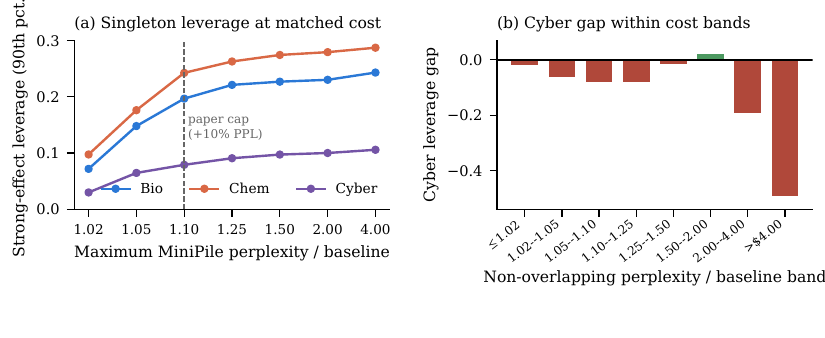}
\caption{\small \textbf{Cyber has lower singleton leverage across the trained rank
ladder.} Each observation gates one MLP coordinate at one rank, all others at full
rank; leverage is the normalised margin drop defined in the text.
(a)~90th-percentile leverage under a cumulative perplexity cap --- a strong, not
single-best, intervention --- with the $1.10\times$ gate dashed. (b)~The same
interventions in disjoint perplexity bands: cyber leverage minus the bio--chem
mean, negative where cyber moves less; the one positive band holds 14 destructive
interventions and alters neither registered comparison. Split~A diagnostic; 756
interventions, uncertainty clustered by transformer block.}
\label{fig:singleton-leverage}
\end{figure}

\textbf{A measurement concern.}
The full-rank model
answers only $43.2\%$ of cyber questions correctly, against $67.6\%$ for bio and
$53.4\%$ for chem. An intervention cannot flip an answer that was already wrong,
so cyber may appear harder to suppress simply because fewer cyber items are known at
baseline. We therefore repeat the comparison using only baseline-correct questions.
To keep the domains comparable, we retain 87 items from each, answer letters
balanced and positions matched in each domain's baseline-margin ordering. For
every in-gate singleton, we ask how much of the matched items' original
decision margin it removes and how often it changes a known answer to a wrong one
(Fig.~\ref{fig:known-item-matching}).

\begin{figure}[htb]
\centering
\includegraphics[width=\linewidth]{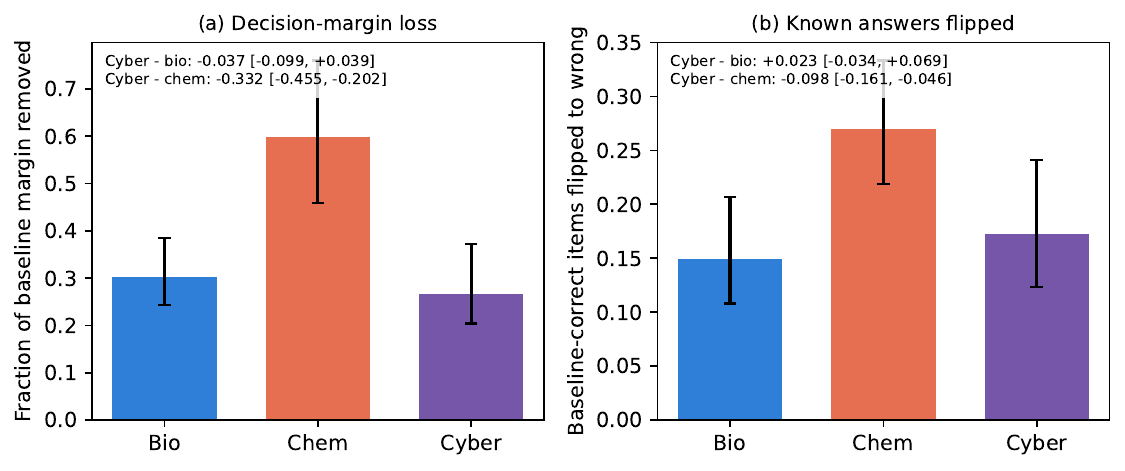}
\caption{\small \textbf{After matching baseline-known questions, cyber responds
like bio; chemistry drives the remaining gap.} Bars: 90th percentile over the 586
in-gate singletons on the matched panel. (a)~Margin loss as a fraction of the
matched domain's baseline margin; (b)~the fraction of baseline-correct answers
flipped. Error bars: 95\% intervals from a crossed bootstrap over transformer
blocks and matched item triplets; insets: pairwise cyber contrasts with 95\%
intervals. Split~A post-hoc diagnostic.}
\label{fig:known-item-matching}
\end{figure}

\textbf{The matched result.}
The interpretation changes. On decision-margin loss, cyber is
$-0.037$ below bio (95\% CI $[-0.099,0.039]$), and on known-answer flips it is
$+0.023$ above bio (95\% CI $[-0.034,0.069]$): neither difference is distinguishable
from zero. Cyber remains below chem on both margin loss ($-0.332$,
95\% CI $[-0.455,-0.202]$) and flips ($-0.098$, 95\% CI
$[-0.161,-0.046]$). Thus lower baseline cyber accuracy explains much of the original
cyber--bio gap, but not the cyber--chem gap.

\textbf{The narrower claim.}
Together, the raw and matched analyses establish a narrower claim than cyber being
intrinsically inseparable. The present MLP-prefix interface has less leverage over
cyber than over the bio--chem average, but this is not a uniform cyber deficit:
baseline-known cyber and bio questions respond similarly, while chemistry is
unusually susceptible. The singleton grid still makes a pure search-failure
explanation unlikely, while leaving four mechanistic possibilities, each with a
separable prediction. \emph{Distributed support} predicts a lower
concentration of positive singleton effects and a gradual catch-up under measured
top-$k$ accumulation. \emph{Rank-order misalignment} predicts cyber sensitivity in
the retained low-index factor bands when non-prefix bands are ablated at matched NLL.
\emph{Computation outside the MLP surface} predicts disproportionate cyber leverage
from attention-head or residual-stream interventions. \emph{Overlap with general
capability} predicts that cyber suppression lies on a worse code/technical-utility
Pareto frontier. The registered WMDP--MMLU correlation test does not support
the broad version of the last account: after removing each intervention's MiniPile
$\Delta$NLL, cyber's mean coupling breadth differs from the bio--chem average by
$-0.032$, 95\% CI $[-0.129,0.053]$. A code-specific retain panel remains the direct
test. Until the band, attention, and accumulation interventions are run, the measured
conclusion is lower interface leverage, not a unique claim about where cyber knowledge
resides.

\begin{figure}[htb]
\centering
\includegraphics[width=\linewidth]{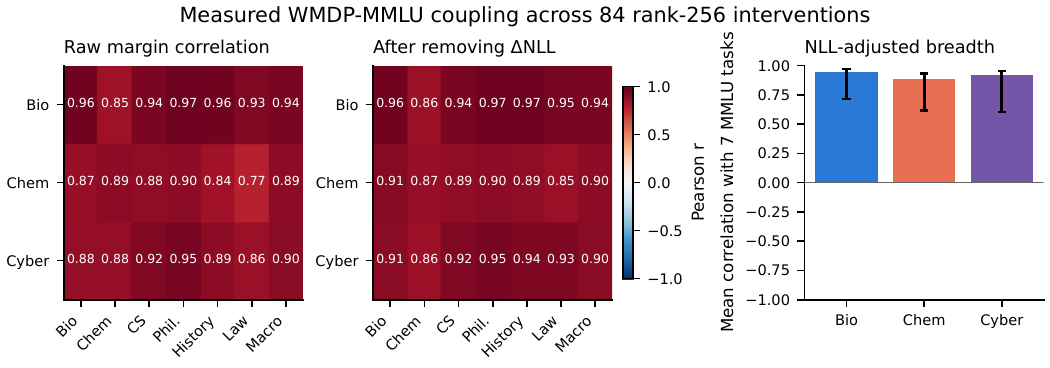}
\caption{\small \textbf{Measured WMDP--MMLU coupling at rank 256.} Pearson
correlations of margin drops across 84 single-coordinate interventions, three WMDP
domains against seven MMLU subjects, before (left) and after (middle) removing each
intervention's MiniPile $\Delta$NLL; right, the mean adjusted correlation with 95\%
block-bootstrap intervals. Coupling is broad for all three domains; the registered
cyber-breadth contrast is indistinguishable from zero. Split~A diagnostic.}
\label{fig:wmdp-mmlu-correlation}
\end{figure}

\FloatBarrier
\subsection{The access matrix}\label{app:access-matrix}

Table~\ref{tab:access-matrix} is the sublattice of Fig.~\ref{fig:sublattice} as a
deployment would consume it: each node a principal against the full metric panel,
on both splits.

\begin{table}[htb]
\centering\small
\caption{\small \textbf{The deployment's access matrix.} Each lattice node as a
principal against the metric panel: per-domain suppression, collateral on both
macros, fluency, and cost as $|\mathrm{supp}(c)|$ with total rank withdrawn.
Measured cells are split~A\,/\,split~B, in points, against that split's own full
privilege (headroom $43.4/40.1$ bio, $27.8/22.8$ chem, $19.2/16.6$\,pp cyber);
fluency and cost are split-independent. Rows ascend in privilege; the three omitted
joins equal $\ctop$ exactly. The single-profile rows are the deployed profiles'
metric panel.}
\label{tab:access-matrix}
\begin{tabular}{@{}lcccccrr@{}}
\toprule
principal & bio & chem & cyber & adj.\ & dist.\ & ppl & dials/rank \\
\midrule
$c_{\mathrm{b}}\wedge c_{\mathrm{ch}}\wedge c_{\mathrm{cy}}$ & 8.6/6.1 & 6.1/9.0 & 2.3/0.3 & 6.4/4.1 & 5.1/3.7 & 1.37$\times$ & 11\,/\,12390 \\
$c_{\mathrm{b}}\wedge c_{\mathrm{ch}}$ & 7.7/6.8 & 9.2/6.9 & 1.7/1.5 & 2.9/2.9 & 3.1/6.0 & 1.18$\times$ & 7\,/\,7782 \\
$c_{\mathrm{b}}\wedge c_{\mathrm{cy}}$ & 8.3/5.0 & 9.2/9.4 & 4.5/1.8 & 3.5/1.7 & 3.1/1.4 & 1.22$\times$ & 8\,/\,8704 \\
$c_{\mathrm{ch}}\wedge c_{\mathrm{cy}}$ & 5.0/3.5 & 4.9/6.9 & 1.4/1.4 & 2.9/1.7 & 4.7/1.8 & 1.27$\times$ & 9\,/\,10342 \\
\midrule
$c_{\mathrm{b}}$ & 7.5/3.8 & 6.7/4.5 & 0.4/0.2 & -2.9/1.7 & 0.4/2.9 & 1.09$\times$ & 4\,/\,4096 \\
$c_{\mathrm{ch}}$ & 5.2/3.9 & 16.0/5.3 & -0.1/0.9 & -1.7/2.3 & 0.2/2.9 & 1.10$\times$ & 5\,/\,5734 \\
$c_{\mathrm{cy}}$ & 1.6/3.8 & 7.4/4.9 & 3.4/-1.0 & -2.3/0.6 & -0.4/-1.9 & 1.12$\times$ & 4\,/\,4608 \\
\midrule
$c_{\mathrm{b}}\vee c_{\mathrm{ch}}$ & 4.9/3.0 & 5.5/1.2 & -1.4/0.8 & -1.7/0.6 & 0.8/4.1 & 1.02$\times$ & 2\,/\,2048 \\
$\ctop$ & 0.0/0.0 & 0.0/0.0 & 0.0/0.0 & 0.0/0.0 & 0.0/0.0 & 1.00$\times$ & -- \\
\bottomrule
\end{tabular}
\end{table}

\FloatBarrier
\subsection{The SmolLM2 generalisation}
\label{app:smollm2generalisation}

\begin{table}[htb]
\centering
\small
\caption{\small \textbf{Cross-family replication.} Each claim in both lineages, on
both halves: $x \to y$ is split~A $\to$ split~B, and suppression is also given as a
fraction of the domain's headroom above chance --- what the objective's floor makes
removable. Split~A reproduces every previously registered value exactly, which is
what licenses reading the split-B column beside it. Provenance differs by
column: the Qwen3 suppression rows quote the granular search arms as previously
registered, not the deployed atoms of Table~\ref{tab:access-matrix} (bio
$7.5 \to 3.8$, chem $16.0 \to 5.3$, cyber $3.4 \to -1.0$); the SmolLM2 rows are
the pre-registered all-selected atoms of Fig.~\ref{fig:sublattice-smollm2}.
Held-out reading costs every
suppression claim roughly half its size in both lineages; chem \emph{inverts} on
SmolLM2; the chem$+$cyber weakening reverses on~B and is withdrawn; utility
non-composition is the only claim that grows on the held-out half, in both.
$\ddagger$:~independently selected parents; a single joint search does reach
in-gate meets. $\S$:~a property of the search,
which never reads split~B. $\P$:~at $1.20\times$ (Qwen3) and $1.30\times$
(SmolLM2); perplexity ratios are split-independent throughout.}
\label{tab:crossfamily}
\begin{tabular}{@{}llll@{}}
\toprule
Claim & Qwen3-1.7B & SmolLM2-1.7B & \\
\midrule
bio suppressible          & $-11.79 \to -8.01$, $27 \to 20\%$ & $-10.06 \to -6.12$, $39 \to 24\%$ & \checkmark \\
chem suppressible         & $-19.02 \to -6.94$, $68 \to 31\%$ & $-7.36 \to \mathbf{+4.08}$, inverts & $\sim$ \\
cyber suppressible$^{\P}$ & $-5.34 \to -1.21$, $28 \to 7\%$ & $-7.85 \to -2.52$, $68 \to 26\%$ & \checkmark \\
no in-gate cyber profile  & cheapest $1.120\times$ & cheapest $>1.10\times$ & \checkmark \\
frontier beats selected$^{\S}$ & 12 of 13 arms & 3 of 3 arms & \checkmark \\
atoms/meets suppress      & $223 \to 186$ of 270 & $44 \to 21$ of 66 & \checkmark \\
no in-gate meet$^{\ddagger}$ & 0 of 24 runs & 0 of 4 runs & \checkmark \\
joins cheap, weak, in-gate & yes & 3 of 3, $1.03$--$1.10\times$ & \checkmark \\
triple strengthens bio    & $-13.21$, 8 $\to$ $-15.38$, 10 runs & survives, 3 runs & \checkmark \\
chem$+$cyber destroys chem & 9 of 24 $\to$ \emph{reverses} & 4 of 4 n.s.\ $\to$ none & --- \\
utility non-composition   & $+6.84 \to +7.00$ distant & $+16.99 \to +17.12$ distant & \checkmark\checkmark \\
\bottomrule
\end{tabular}
\end{table}

\begin{figure}[htb]
\centering
\includegraphics{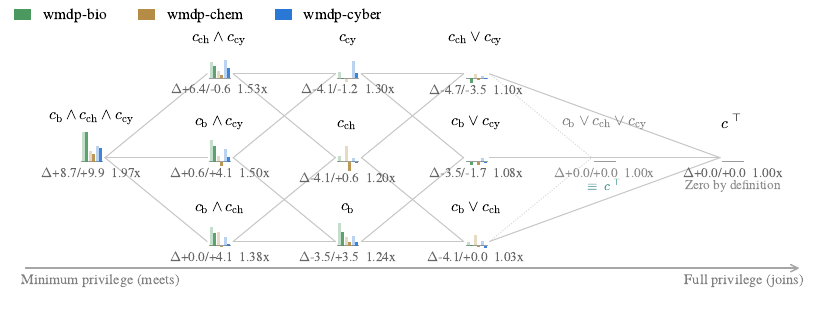}
\vspace{-10pt}
\caption{\small \textbf{The measured sublattice (SmolLM2-1.7B).}
The all-selected run's twelve nodes, drawn exactly as Fig.~\ref{fig:sublattice}
and pinned to its $16$\,pp scale, so heights are directly comparable across
lineages. chem's atom bar runs \emph{below} baseline on~B: the gate raises chem
accuracy in this lineage. The triple strengthens bio; joins are cheap, weak, and
the only in-gate nodes ($1.03$--$1.10\times$), the faded one equal to $\ctop$
exactly, its parents gating disjoint supports. The distant macro, not drawn, is
Theorem~\ref{thm:nobound}'s demonstration in this lineage: no atom moves it by
more than $1.6$\,pp, the triple loses $17.1$ --- individually harmless, jointly
destructive.}
\label{fig:sublattice-smollm2}
\vspace{-10pt}
\end{figure}

\begin{figure}[!htb]
\centering
\includegraphics[width=0.5122\linewidth]{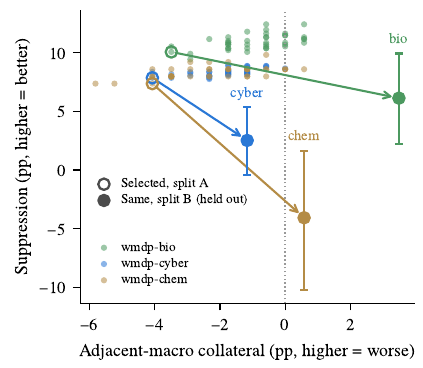}\hfill
\includegraphics[width=0.4878\linewidth]{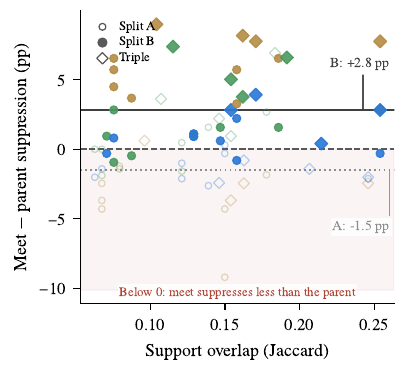}
\vspace{-10pt}
\caption{\small \textbf{Profiles and composition (SmolLM2-1.7B).}
Figs.~\ref{fig:profiles-and-composition-a}
and~\ref{fig:profiles-and-composition-b} redrawn on SmolLM2 from this appendix's
three search arms and four lattice runs. (a)~The $120$ gate-verified finalists,
$40$ per arm. SmolLM2 has no probe run, so the selected markers are read off the
all-selected lattice run's atoms, whose gated supports are identical to the arms'
pre-registered picks --- the values of Table~\ref{tab:crossfamily}'s first three
rows. Every arrow points down and to the right (less suppression, more
collateral), only bio's held-out interval clears zero, and collateral changes
sign for bio and chem: gating \emph{helped} the adjacent macros on the half that
chose it. (b)~All $36$ (meet, parent, domain) records from the four runs. The
reversal reproduces at a quarter of
Qwen3's family size: sub-additive records fall from $25/36$ on~A to $17/36$ on~B, and
the medians move $-1.49 \to +2.82$\,pp ($\eta$ $0.84 \to 1.14$).}
\label{fig:profiles-and-composition-smollm2}\sublabel{fig:profiles-and-composition-smollm2-a}\sublabel{fig:profiles-and-composition-smollm2-b}
\vspace{-10pt}
\end{figure}

\textbf{Design.}
Table~\ref{tab:crossfamily}'s lattice rows pool four SmolLM2 lattice runs. Because each
search arm re-measured $40$ finalists on the reporting set, two parents per domain are
available without any further search: the arm's own pre-registered pick, and the
strongest finalist that leaves both retention macros within $2$\,pp. Crossing two levels of three factors in a
$2^{3-1}$ fractional factorial (bio $\oplus$ chem $\oplus$ cyber $= 0$) covers all twelve
distinct parent pairs in four runs, of which the all-selected run is one, so three
additional runs suffice. Each run emits twelve nodes (three atoms, three pairwise meets,
the triple, three joins, and the triple join) in about seven minutes.

Parents. The all-selected run's three atoms are the nodes
$c_{\mathrm{b}}$, $c_{\mathrm{ch}}$ and $c_{\mathrm{cy}}$ of
Fig.~\ref{fig:sublattice-smollm2}. The frontier parents, which no figure draws, are bio
$-12.42$\,pp at $1.3831\times$, chem $-9.82$ at $1.2225$ and cyber $-8.96$ at $1.2772$. A
frontier parent is selected on raw suppression rather than on the
pre-registered bootstrapped lower quartile, so it is less robust by construction (stable
scores $-2.996 / -1.942 / -1.407$ against $-1.686 / -0.166 / -0.104$); each result records
this per domain in \texttt{profile\_provenance.post\_hoc\_frontier\_pick}, and the
all-selected run carries no such asterisk and is the one quoted in the main text. The
profile-reproduction guard was exact on all four runs (worst deviation $0.0000$\,pp of a
$0.05$ tolerance, over $36$ checks), so every parent transferred into the lattice without
drift.

Corrected verdict over the pooled family: composition, $36$ tests, BH keeps one --- the
triple strengthening bio ($-6.92$\,pp, $p = 10^{-4}$); suppression, $84$ tests deduplicated
to $66$, BH keeps $44$. The chem$+$cyber weakening of chem appears at $+4.29$, $+4.29$,
$+5.52$ and $+9.20$\,pp across the four runs. Four-of-four agreement in sign is not itself a
test at this family size (a sign test alone gives $p = 0.125$), which is why the main text
reports the row as consistent and unresolved.

\FloatBarrier
\end{document}